\documentclass[final,3p,times,fleqn]{elsarticle}
\usepackage{mathtools,booktabs}
\usepackage{bm}
\usepackage{autobreak}
\usepackage{amsmath}
\usepackage{amsthm}
\usepackage{amssymb}
\usepackage{mathrsfs}
\usepackage{siunitx}
\usepackage{gensymb}
\usepackage{subfigure}
\usepackage{epstopdf}
\usepackage{color}
\usepackage{hyperref}
\usepackage{numprint}
\usepackage{CJKutf8}

\biboptions{sort&compress}
\begin{document}
\begin{CJK*}{UTF8}{gbsn}
\begin{frontmatter}
	\title{A thermodynamically consistent and conservative diffuse-interface model for gas-liquid-solid multiphase flows}
	\author[a,b,c]{Chengjie Zhan}
	\author[a,b,c]{Xi Liu}
	\author[a,b,c]{Zhenhua Chai \corref{cor1}}
	\ead{hustczh@hust.edu.cn}
	\author[a,b,c]{Baochang Shi}
	\address[a]{School of Mathematics and Statistics, Huazhong University of Science and Technology, Wuhan 430074, China}
	\address[b]{Institute of Interdisciplinary Research for Mathematics and Applied Science, Huazhong University of Science and Technology, Wuhan 430074, China}
	\address[c]{Hubei Key Laboratory of Engineering Modeling and Scientific Computing, Huazhong University of Science and Technology, Wuhan 430074, China}
	\cortext[cor1]{Corresponding author.}
	\begin{abstract} 
		In this work, a thermodynamically consistent and conservative diffuse-interface model for gas-liquid-solid multiphase flows is proposed. In this model, a novel free energy for the gas-liquid-solid multiphase flows is established according to a ternary phase-field model, and it not only contains the standard bulk and interface free energies for two-phase flows, but also includes some additional terms to reflect the penalty in the solid phase and the wettability on the solid surface. Furthermore, a smooth indicator function of the solid phase is also introduced in the consistent Navier-Stokes equations to achieve a high viscosity in the solid phase and preserve the velocity boundary conditions on the solid surface. Based on the proposed diffuse-interface model, the fluid interface dynamics, the fluid-structure interaction, and the wetting property of the solid surface can be described simply and efficiently. Additionally, the total energy is also proved to be dissipative for the two-phase flows in the stationary geometries. To test the present diffuse-interface model, we develop a consistent and conservative lattice Boltzmann method and conduct some simulations. The numerical results also confirm the energy dissipation and good capability of the proposed diffuse-interface model in the study of two-phase flows in complex geometries and gas-liquid-particle multiphase flows.    
	\end{abstract}
\end{frontmatter}		
\end{CJK*}
\section{Introduction}
Gas-liquid-solid multiphase flows are ubiquitous in nature and many industrial processes, such as the rain drops impacting on leaves and the ships ploughing through sea waves \cite{Lloyd1989}. 
With the development of computer science and scientific computing, the numerical simulation has become an effective tool to reveal the principles of such multiphase flows and to provide guidance on industrial manufacturing processes. 

For the complex gas-liquid-solid multiphase flows, there are several common challenges in numerical simulations, i.e., the fluid interfacial dynamics, the fluid-structure interaction (FSI) with moving boundaries, and the contact line dynamics.
To track or capture the fluid interface in multiphase flows, there are two categories of methods: the sharp-interface method \cite{Unverdi1992JCP,Sussman2007JCP} and the diffuse-interface method \cite{Anderson1998ARFM,Jacqmin1999JCP}. As a diffuse-interface approach, the phase-field model (PFM) \cite{Jacqmin1999JCP,Badalassi2003JCP} has been widely used in the study of the fluid interfacial dynamics in multiphase flows through the fourth-order Cahn-Hilliard (CH) equation \cite{Cahn1996EJAM} or the conservative second-order Allen-Cahn (AC) equation \cite{Chiu2011JCP} due to the conservation of mass and the no need of explicit interface-tracking. 
Furthermore, the FSI usually plays an important role in the multiphase flows with moving solid objects, such as particle sedimentation, dynamics of fluttering and tumbling plates \cite{Andersen2005JFM}. To accurately resolve the FSI problems, a popular numerical tool is the immersed boundary method (IBM) \cite{Peskin1977JCP,Fogelson1988JCP}, in which the solid boundary is marked by a set of Lagrangian points, and the FSI between the immersed boundary and the neighboring fluid is depicted by introducing an extra force term with the Dirac delta function. Recently, some other effective methods have also been proposed for the FSI problems, for instance, the partially saturated method \cite{Noble1998IJMPC,Cook2004EC}, the volume penalization method \cite{Angot1999NM,Kolomenskiy2009JCP}, and the smoothed profile method (SPM) \cite{Nakayama2005PRE,Nakayama2008EPJE}. In the SPM, a smooth indicator function is introduced to represent the local volume fraction of the solid phase, and the extra force included in the Navier-Stokes (NS) equations is similar to that in the IBM.  
Additionally, the wettability of the solid surface is also needed to be considered in the gas-liquid-solid multiphase flows, especially for the case with a small capillary number. To describe the wetting property on solid surface, a wetting boundary condition was proposed from the viewpoint of the geometrical relation \cite{Ding2007PRE}, and it has been widely used in the past years \cite{Sui2014ARFM,Li2016PECS}. On the other hand, in the PFM, the wetting boundary condition can be constructed through introducing an additional wall free energy functional \cite{Gennes1985EMP,Moldover1980}, which also shows a good performance \cite{Huang2015IJNMF,Liang2019PRE}. However, when the two-phase flows in complex geometries are considered, the discretization of the wetting boundary condition is more complicated. To avoid the direct treatment of wetting boundary condition at complex boundaries, some diffuse domain approaches \cite{Aland2010CMES,Guo2021JFM,Yang2023JCP} have been developed in recent years. In these methods, as shown in Fig. \ref{fig-domain}, the boundary conditions at the boundary $\Gamma$ are treated as the source terms in the governing equations on a larger and regular computational domain $\Omega$. However, the reconstruction of the governing equations or numerical schemes is also an important issue, and some additional treatments are needed. More recently, Hong et al. \cite{Hong2023JCP} proposed a phase-field approach to investigate the two-phase flows in complex domains with moving contact lines. In their model, the original boundary integral of wall free energy density for the wettability at the fluid-solid boundary is rewritten as a volume integral in the larger rectangular domain by considering an approximation of the surface delta function, and an additional penalty term is also introduced into the free energy to preserve the mass conservation of the AC equation. However, in this model, the value of the penalty parameter needs to be determined through some numerical tests, and the velocity boundary conditions at the fluid-solid boundary are not preserved in the rectangular domain.

To address the above issues for complex gas-liquid-solid flows and consider the FSI and the wettability simultaneously, some diffuse-interface models have been proposed. Shinto \cite{Shinto2012APT} presented a phase-field free energy functional for fluid domain and colloid domain, and an additional term is introduced to describe the wetting surface of colloids. In this model, the positive parameter in the additional term must be chosen to be a large value such that the free energy density in colloids exhibits a single-well profile, and the minimum value of the single-well function can control the wettability of the colloid surface. Based on the similar idea, Lecrivain et al. \cite{Lecrivain2016PF,Lecrivain2017PRE,Lecrivain2018PRF,Lecrivain2020JCP} developed some phase-field models to simulate the attachment of particles, the arbitrarily shaped particle at the fluid interface, and the elasto-capillary deformation of a flexible fiber. Additionally, Mino et al. \cite{Mino2020PRE,Mino2022PRE} coupled the free energy lattice Boltzmann (LB) method with the improved SPM, and the minimum value of the single-well profile controlling the wettability can be given by a theoretical estimation \cite{Iwahara2003}. More recently, Zheng et al. \cite{Zheng2023PRE} developed a phase-field based LB equation to simulate gas-liquid-particle fluid dynamics together with the SPM. In their model, the SPM is applied to both the conservative AC equation and the incompressible NS equations, and the contact angle is also determined by an affinity parameter through a cubic polynomial \cite{Iwahara2003}. However, to preserve the mass conservation of the modified AC equation, a remedial strategy should be adopted. 

In this paper, we will propose a unified diffuse-interface model for the gas-liquid-solid multiphase flows, where the fluid interfacial dynamics, the FSI, and the wettability of solid surface can be described simply and efficiently. Based on a classical ternary PFM, we develop a novel free energy functional for the gas-liquid-solid multiphase flows, and the free energy functional can be divided into three parts: the first one is the standard form for two-phase flows, the second one is a penalty for the phase-field variable in solid phase, and the third one can reflect the wetting property of solid surface. It is worth noting that the contact angle is explicitly included in the third part of the free energy, instead of using a user-defined control parameter in the previous models \cite{Shinto2012APT,Lecrivain2020JCP,Zheng2023PRE}, and we can show that under some proper approximations, the third part is also consistent with the previous wall energy for two-phase flows \cite{Jacqmin2000JFM,Yue2010JFM}. Through minimizing the present free energy functional, one can obtain the conservative fourth-order CH equation for interface capturing. In addition, the FSI between fluid and free solid object can be depicted by a simplified form of the diffuse-interface method \cite{Liu2022CF}, and it has the similar form with the IBM and SPM at the leading order.

The rest of this paper is organized as follows. In Section \ref{models}, according to a ternary free energy functional, a thermodynamically consistent and conservative PFM for gas-liquid-solid multiphase flows is first proposed, followed by the consistent and conservative NS equations for fluid flows. Then the proposed phase-field-NS system is proved to possess an energy dissipation law for two-phase flows in the stationary geometries, and for the multiphase flows with free rigid particles, the governing equations for particle motion are also introduced. The numerical methods are developed in Section \ref{methods}, where the LB method is the primary numerical tool for partial differential equations due to its features of simplicity in coding and fully parallel algorithm \cite{Higuera1989EL,Chen1998ARFM,Aidun2010ARFM}. In Section \ref{Simulations}, several typical benchmark problems are adopted to test the proposed diffuse-interface model, and the total energy is also numerically confirmed to be dissipative for two-phase flows in stationary geometries. Furthermore, the present diffuse-interface model is applied to study some complex gas-liquid-solid multiphase flows. Finally, some conclusions are summarized in Section \ref{Conclusions}.

\begin{figure}
	\centering
	\includegraphics[width=2.0in]{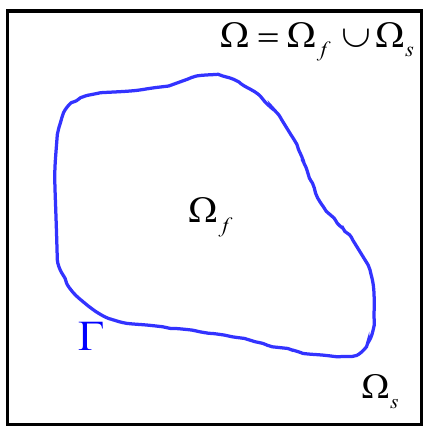}
	\caption{The computational domain $\Omega$ is the union of the fluid domain $\Omega_f$ and the solid domain $\Omega_s$, and the interface between these two domains is $\Gamma$.}
	\label{fig-domain}
\end{figure}

\section{Mathematical model for gas-liquid-solid multiphase flows}\label{models}
\subsection{Phase-field model for interface capturing}
In the phase-field theory for two-phase flows, the free energy density can be given by \cite{Jacqmin1999JCP}
\begin{equation}\label{eq-2energyDensity}
	f_1\left(\phi,\nabla\phi\right)=\frac{3\sigma_{12}}{4D}\left(\phi-1\right)^2\left(\phi+1\right)^2+\frac{3D\sigma_{12}}{16}|\nabla\phi|^2,
\end{equation}
where the order parameter $\phi$ is smoothly changed from 1 in phase 1 to -1 in phase 2, $D$ is the interface thickness, and $\sigma_{12}$ is the surface tension coefficient between phase 1 and phase 2.
Through minimizing the free energy functional, one can obtain the following CH equation,
\begin{equation}\label{eq-CHE}
	\frac{\partial\phi}{\partial t}+\nabla\cdot\left(\phi\mathbf{u}\right)=\nabla\cdot M\nabla\mu_{\phi},
\end{equation}
where $M$ is the mobility, $\mathbf{u}$ is the velocity of fluid, $\mu_{\phi}$ is the chemical potential and can be obtained by applying the variational operator to the free energy functional. 

Furthermore, to extend the PFM to depict multiphase systems, Boyer and Lapuerta \cite{Boyer2006MMNA} elaborately designed the following mix free energy for three-component flows, 
\begin{equation}\label{eq-3energy}
	E=\int_{\Omega}\left[F\left(\phi_0,\phi_1,\phi_2\right)+\frac{3D}{8}\sum_{k=0}^{2}\gamma_k|\nabla\phi_k|^2\right]\mathrm{d}\Omega,
\end{equation}
where $\phi_k\in[0,1]$ represents the volume fraction of phase $k$ and satisfies $\phi_0+\phi_1+\phi_2=1$. The spreading coefficient $\gamma_k$ is related to the pairwise symmetric surface tension coefficient $\sigma_{kl}$,
\begin{equation}
	\begin{aligned}
		\gamma_0&=\sigma_{01}+\sigma_{02}-\sigma_{12},\\
		\gamma_1&=\sigma_{01}+\sigma_{12}-\sigma_{02},\\
		\gamma_2&=\sigma_{02}+\sigma_{12}-\sigma_{01},
	\end{aligned}
\end{equation}
and the bulk free energy density is given by \cite{Boyer2006MMNA}
\begin{equation}
	F\left(\phi_0,\phi_1,\phi_2\right)=\frac{12}{D}\left[\frac{\gamma_0}{2}\phi_0^2\left(1-\phi_0\right)^2+\frac{\gamma_1}{2}\phi_1^2\left(1-\phi_1\right)^2+\frac{\gamma_2}{2}\phi_2^2\left(1-\phi_2\right)^2\right].
\end{equation}

In this work, we introduce a phase-field variable $\phi=\phi_1-\phi_2$, and then obtain the formulas $\phi_1=\left(1-\phi_0+\phi\right)/2$ and $\phi_2=\left(1-\phi_0-\phi\right)/2$. Based on these two relations, the three-component mix free energy (\ref{eq-3energy}) can be expressed by $\phi_0$, $\phi$, and their gradients,
\begin{equation}
	E=E_0+E_1+E_2,\quad E_1=\int_{\Omega}f_1\left(\phi,\nabla\phi\right)\mathrm{d}\Omega,\quad E_2=\int_{\Omega}f_2\left(\phi_0,\phi,\nabla\phi_0,\nabla\phi\right)\mathrm{d}\Omega,
\end{equation}
where $E_0$ is only the function of $\phi_0$, $E_1$ is related to $\phi$ and exactly is the integral of the free energy density (\ref{eq-2energyDensity}), while $f_2\left(\phi_0,\phi,\nabla\phi_0,\nabla\phi\right)$ contains the cross terms, and is given by
\begin{equation}
	f_2\left(\phi_0,\phi,\nabla\phi_0,\nabla\phi\right)=\frac{9\sigma_{12}}{2D}\phi_0^2\phi^2+\frac{3}{D}\left(\sigma_{02}-\sigma_{01}\right)\phi_0\phi\left(\phi^2+\phi_0^2-1\right)+\frac{3D}{8}\left(\sigma_{02}-\sigma_{01}\right)\nabla\phi_0\cdot\nabla\phi.
\end{equation}

According to the above reformulation of the three-component free energy, we can set $\phi_0$ to be the smooth indicator function of the solid phase when considering the gas-liquid-solid multiphase flows, and develop a novel free energy for the total system where the fluid and solid domains are included (see the diagram in Fig. \ref{fig-domain}),
\begin{equation}\label{eq-2energy}
	\begin{aligned}
		\mathcal{E}_F&=\int_{\Omega}\left[f_1\left(\phi,\nabla\phi\right)+f_2\left(\phi_0,\phi,\nabla\phi_0,\nabla\phi\right)\right]\mathrm{d}\Omega\\
		&=\int_{\Omega}\left[\frac{3\sigma_{12}}{4D}\left(1-\phi\right)^2\left(1+\phi\right)^2+\frac{3D\sigma_{12}}{16}|\nabla\phi|^2+\frac{9\sigma_{12}}{2D}\phi_0^2\phi^2+\frac{3\sigma_{12}\cos\theta}{D}\phi_0\phi\left(\phi^2+\phi_0^2-1\right)+\frac{3D\sigma_{12}\cos\theta}{8}\nabla\phi_0\cdot\nabla\phi\right]\mathrm{d}\Omega,
	\end{aligned}
\end{equation}
where the Young's law \cite{Young1805} is adopted, i.e., $\sigma_{02}-\sigma_{01}=\sigma_{12}\cos\theta$ with $\theta$ being the contact angle of phase 1 to the solid phase, as shown in Fig. \ref{fig-laplace}. 
Here we note that the proposed free energy can reduce to the standard one for binary flows when solid phase $\phi_0$ disappears. Actually, the free energy density in Eq. (\ref{eq-2energy}) can be divided into three parts: the first two terms are the original bulk and interface parts for two-phase flows, the third term can be seen as the penalty of $\phi=0$ when $\phi_0=1$, and the last two terms are related to the wettability on the solid surface. 

Then the chemical potential can be given by the following variational
form,
\begin{equation}\label{eq-potential}
	\mu_{\phi}=\frac{\delta \mathcal{E}_F}{\delta\phi}=\frac{3\sigma_{12}}{D}\phi\left(\phi-1\right)\left(\phi+1\right)-\frac{3D\sigma_{12}}{8}\nabla^2\phi+\frac{9\sigma_{12}}{D}\phi_0^2\phi+\frac{3\sigma_{12}\cos\theta}{D}\phi_0\left(3\phi^2+\phi_0^2-1\right)-\frac{3D\sigma_{12}\cos\theta}{8}\nabla^2\phi_0.
\end{equation}
With the above modified chemical potential $\mu_{\phi}$, the CH equation (\ref{eq-CHE}) can be used to describe the gas-liquid-solid multiphase flows with different wetting properties.

\begin{figure}
	\centering
	\includegraphics[width=2.4in]{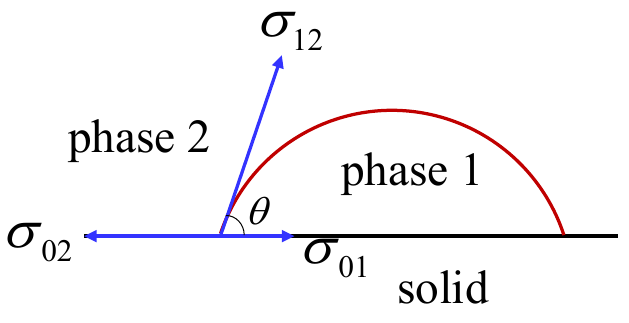}
	\caption{Three-phase contact diagram.}
	\label{fig-laplace}
\end{figure}

In addition, to make a clear explanation of the energy density terms in Eq. (\ref{eq-2energy}), we first define a new bulk free energy density,
\begin{equation}\label{eq-energybulk}
	f_b\left(\phi,\phi_0\right)=\frac{3\sigma_{12}}{4D}\left(1-\phi\right)^2\left(1+\phi\right)^2+\frac{9\sigma_{12}}{2D}\phi_0^2\phi^2,
\end{equation}
and plot it in Fig. \ref{fig-energy}. From this figure, one can see that when $\phi_0=0$ in fluid domain, the bulk free energy density $f_b$ reduces to the original double-well function with two minima located at $\phi=\pm1$, while it turns into a single-well profile with the minimum located at $\phi=0$ in solid phase when $\phi_0=1$. These results indicate that the last term in Eq. (\ref{eq-energybulk}) can be considered as a penalty.

Furthermore, we also present an analysis on the part of free energy related to the wettability. We now define the following free energy,  
\begin{equation}\label{eq-energywall}
	\mathcal{E}_w=\int_{\Omega}\left[\frac{3\sigma_{12}\cos\theta}{D}\phi_0\phi\left(\phi^2+\phi_0^2-1\right)+\frac{3D\sigma_{12}\cos\theta}{8}\nabla\phi_0\cdot\nabla\phi\right]\mathrm{d}\Omega,
\end{equation}
when the variational operator is used to the above free energy (\ref{eq-energywall}), one can obtain
\begin{equation}\label{eq-DEi}
	\delta\mathcal{E}_w=\int_{\Omega}\left[\frac{3\sigma_{12}\cos\theta}{D}\phi_0\left(3\phi^2+\phi_0^2-1\right)-\frac{3D\sigma_{12}\cos\theta}{8}\nabla^2\phi_0\right]\delta\phi\mathrm{d}\Omega,
\end{equation}
where the condition $\mathbf{n}_{\Omega}\cdot\nabla\phi_0=0$ is used. Additionally, we can suppose the indicator function of solid phase $\phi_0$ to be always at the equilibrium state and satisfy a hyperbolic tangent profile $\phi_0=\left[1+\tanh\left(2\zeta/D\right)\right]/2$, where $\zeta$ is the signed distance function, then we can get the following relations,
\begin{equation}\label{eq-phi0Eq}
	\nabla^2\phi_0=\frac{16}{D^2}\phi_0\left(1-\phi_0\right)\left(1-2\phi_0\right),\quad |\nabla\phi_0|=\frac{4\phi_0\left(1-\phi_0\right)}{D}.
\end{equation} 
Substituting Eq. (\ref{eq-phi0Eq}) into Eq. (\ref{eq-DEi}) yields
\begin{equation}
	\delta\mathcal{E}_w=\int_{\Omega}\frac{3\sigma_{12}\cos\theta}{4}|\nabla\phi_0|\left[\frac{3\phi^2}{1-\phi_0}-3\left(1-\phi_0\right)\right]\delta\phi\mathrm{d}\Omega.
\end{equation}

\begin{figure}
	\centering
	\includegraphics[width=3.5in]{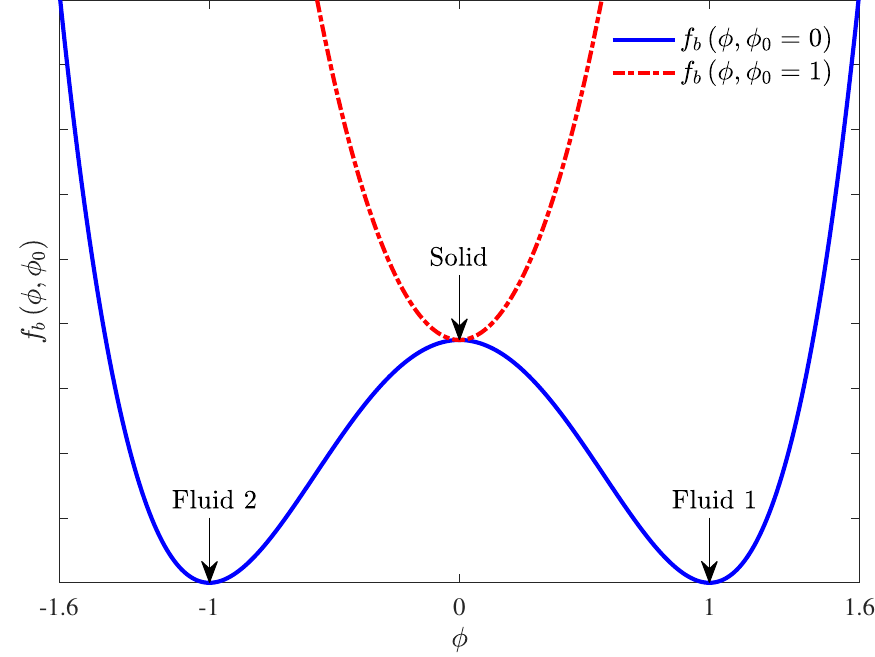}
	\caption{The bulk free energy density.}
	\label{fig-energy}
\end{figure}

Usually, the surface delta function of the fluid-solid interface can be approximated by $|\nabla\phi_0|$ in a diffuse-interface method, and satisfies $\int_{\Omega}|\nabla\phi_0|h\mathrm{d}\Omega=\int_{\Gamma}h\mathrm{d}A$ for any smooth function $h$.
For this reason, the free energy $\mathcal{E}_w$ in $\Omega$ can be written as a boundary integral at the fluid-solid interface $\Gamma$,
\begin{equation}
	\mathcal{E}_w=3\int_{\Gamma}\psi_w\left(\phi,\phi_0\right)\mathrm{d}A,\quad \psi_w\left(\phi,\phi_0\right)=\frac{\sigma_{01}+\sigma_{02}}{2}-\frac{\sigma_{12}}{4}\left[3\phi\left(1-\phi_0\right)-\frac{\phi^3}{1-\phi_0}\right]\cos\theta.
\end{equation}
Here one can observe that when $\phi_0=0$, the above energy density $\psi_w$ is consistent with the standard wall energy function for wettability in two-phase flows \cite{Jacqmin2000JFM,Yue2010JFM}. In other words, the free energy (\ref{eq-energywall}) of present diffuse-interface model can accurately control the wettability on solid surface.  

\subsection{The incompressible Navier-Stokes equations for fluid flows}
To describe the fluid flows, the following consistent and conservative NS equations are applied,
\begin{subequations}\label{eq-NSE}
	\begin{equation}\label{eq-NS1}
		\nabla\cdot\mathbf{u}=0,
	\end{equation}
	\begin{equation}\label{eq-NS2}
		\frac{\partial\left(\rho\mathbf{u}\right)}{\partial t}+\nabla\cdot\left[\left(\rho\mathbf{u}-\mathbf{S}\right)\mathbf{u}\right]=-\nabla P+\nabla\cdot\frac{\mu}{1-\phi_0}\left[\nabla\mathbf{u}+\left(\nabla\mathbf{u}\right)^\top\right]+\mu_{\phi}\nabla\phi+\rho\mathbf{f}+\mathbf{F}_b,
	\end{equation}
\end{subequations}
where $P$ is the pressure, the dynamic viscosity $\mu$ is modified by including the term of $1-\phi_0$ to achieve a high viscosity in the solid phase with $\phi_0=1$ and a normal viscosity in fluid region when $\phi_0=0$. 
The fluid density $\rho$ and the viscosity $\mu$ are usually related to the order parameter by
\begin{equation}\label{eq-rho}
	\rho=\frac{\rho_1-\rho_2}{2}\phi+\frac{\rho_1+\rho_2}{2},
\end{equation}
\begin{equation}
	\mu=\frac{\mu_1-\mu_2}{2}\phi+\frac{\mu_1+\mu_2}{2},
\end{equation}
where $\rho_1$, $\mu_1$ as well as $\rho_2$ and $\mu_2$ are the material properties of the pure fluid 1 and fluid 2, respectively. The mass diffusion flux between two phases is given by $\mathbf{S}=M\nabla\mu_{\phi}\left(\rho_1-\rho_2\right)/2$ according to the phase-field equation (\ref{eq-CHE}), $\mu_{\phi}\nabla\phi$ is the surface tension force, $\rho\mathbf{f}$ is the force caused by the FSI, and $\mathbf{F}_b$ is other body force, such as the gravity. We would like to point out that with above definition of the consistent mass flux, the conservative NS equations (\ref{eq-NSE}) satisfy the consistency of reduction, the consistency of mass and momentum transport, and the consistency of mass conservation \cite{Abels2012MMMAS,Huang2020JCP,Mirjalili2021JCP}.

In addition, it should be noted that when the solid object is free moving at the fluid interface, such as gas-liquid-particle multiphase flows, an additional force $\mu_{\phi_0}\nabla\phi_0$ should be included with $\mu_{\phi_0}$ being given by
\begin{equation}\label{eq-muPhi0}
	\mu_{\phi_0}=\frac{\delta \mathcal{E}_F}{\delta\phi_0}=\frac{9\sigma_{12}}{D}\phi_0\phi^2+\frac{3\sigma_{12}\cos\theta}{D}\phi\left(3\phi_0^2+\phi^2-1\right)-\frac{3D\sigma_{12}\cos\theta}{8}\nabla^2\phi.
\end{equation} 
Here one can also find from Eq. (\ref{eq-muPhi0}) that for the particle flows in a single-phase homogeneous medium, i.e., $\phi=0$, above chemical potential $\mu_{\phi_0}$ becomes zero and thus the additional force $\mu_{\phi_0}\nabla\phi_0$ disappears. 

\subsection{The energy dissipation law}
Now we will show the second law of thermodynamics of the phase-field-NS system for the two-phase flows in the stationary geometries. The second law of thermodynamics states that the total energy of an isothermal multiphase system, including both the free energy and kinetic energy, is not increasing without any external energy input.

We denote the total energy of the system to be $\mathcal{E}_{total}=\mathcal{E}_F+\mathcal{E}_K$, where $\mathcal{E}_K=\int_{\Omega}e_K\mathrm{d}\Omega$ is the kinetic energy with energy density defined by $e_K=\rho\mathbf{u}^2/2$.

\newtheorem{theorem}{Theorem}
\begin{theorem}\label{Theorem}
	The phase-field-NS system [Eqs. (\ref{eq-CHE}), (\ref{eq-potential}) and (\ref{eq-NSE})] possesses the following energy dissipation law,
	\begin{equation}
		\frac{\mathrm{d}\mathcal{E}_{total}}{\mathrm{d}t}=-\int_{\Omega}\left\{M\nabla\mu_{\phi}\cdot\nabla\mu_{\phi}+\frac{\mu}{2\left(1-\phi_0\right)}\left[\nabla\mathbf{u}+\left(\nabla\mathbf{u}\right)^\top\right]:\left[\nabla\mathbf{u}+\left(\nabla\mathbf{u}\right)^\top\right]\right\}\mathrm{d}\Omega\leq0.
	\end{equation}
\end{theorem} 
\begin{proof}
	Performing the dot product between $\mathbf{u}$ and the momentum equation (\ref{eq-NS2}), one can obtain the following expression, 
	\begin{equation}
		\begin{aligned}
			&\mathbf{u}\cdot\left\{\frac{\partial\left(\rho\mathbf{u}\right)}{\partial t}+\nabla\cdot\left[\left(\rho\mathbf{u}-\mathbf{S}\right)\mathbf{u}\right]\right\}=\frac{\partial e_k}{\partial t}+\nabla\cdot\left[\left(\rho\mathbf{u}-\mathbf{S}\right)\frac{\mathbf{u}\cdot\mathbf{u}}{2}\right]+\frac{\mathbf{u}\cdot\mathbf{u}}{2}\left[\frac{\partial\rho}{\partial t}+\nabla\cdot\left(\rho\mathbf{u}-\mathbf{S}\right)\right]\\
			&=-\nabla\cdot\left(P\mathbf{u}\right)+\nabla\cdot\left\{\frac{\mu}{1-\phi_0}\left[\nabla\mathbf{u}+\left(\nabla\mathbf{u}\right)^\top\right]\cdot\mathbf{u}\right\}-\frac{\mu}{2\left(1-\phi_0\right)}\left[\nabla\mathbf{u}+\left(\nabla\mathbf{u}\right)^\top\right]:\left[\nabla\mathbf{u}+\left(\nabla\mathbf{u}\right)^\top\right]+\mu_{\phi}\nabla\phi\cdot\mathbf{u}+\left(\rho\mathbf{f}+\mathbf{F}_b\right)\cdot\mathbf{u}.
		\end{aligned}
	\end{equation}
	Here the last term on the right-hand side of the first line in the above equation is zero according to the mass conservation from CH equation. Therefore, the variation of the kinetic energy can be given by
	\begin{equation}\label{eq-DtEk}
		\begin{aligned}
			\frac{\mathrm{d}\mathcal{E}_K}{\mathrm{d}t}=\int_{\Omega}\frac{\partial e_k}{\partial t}\mathrm{d}\Omega=&\int_{\partial\Omega}\left\{\left(\rho\mathbf{u}-\mathbf{S}\right)\frac{\mathbf{u}\cdot\mathbf{u}}{2}-P\mathbf{u}+\frac{\mu}{1-\phi_0}\left[\nabla\mathbf{u}+\left(\nabla\mathbf{u}\right)^\top\right]\cdot\mathbf{u}\right\}\cdot\mathrm{d}\mathbf{s}+\int_{\Omega}\left(\rho\mathbf{f}+\mathbf{F}_b\right)\cdot\mathbf{u}\mathrm{d}\Omega\\
			&+\int_{\Omega}\mu_{\phi}\nabla\phi\cdot\mathbf{u}\mathrm{d}\Omega-\int_{\Omega}\frac{\mu}{2\left(1-\phi_0\right)}\left[\nabla\mathbf{u}+\left(\nabla\mathbf{u}\right)^\top\right]:\left[\nabla\mathbf{u}+\left(\nabla\mathbf{u}\right)^\top\right]\mathrm{d}\Omega.
		\end{aligned}
	\end{equation}
	
	Similarly, the following relations can be derived by multiplying $\mu_{\phi}$ and $\partial_t\phi$ to the CH equation (\ref{eq-CHE}) and the chemical potential (\ref{eq-potential}), respectively, 
	\begin{subequations}
		\begin{equation}
			\mu_{\phi}\partial_t\phi=-\mu_{\phi}\mathbf{u}\cdot\nabla\phi-M\nabla\mu_{\phi}\cdot\nabla\mu_{\phi}+\nabla\cdot\left(M\mu_{\phi}\nabla\mu_{\phi}\right),
		\end{equation}
		\begin{equation}
			\begin{aligned}
				\mu_{\phi}\partial_t\phi=&\partial_t\phi\left[\frac{3\sigma_{12}}{D}\phi\left(\phi-1\right)\left(\phi+1\right)+\frac{9\sigma_{12}}{D}\phi_0^2\phi+\frac{3\sigma_{12}\cos\theta}{D}\phi_0\left(3\phi^2+\phi_0^2-1\right)\right]\\
				&-\frac{3D\sigma_{12}}{8}\nabla\cdot\left[\partial_t\phi\left(\nabla\phi+\cos\theta\nabla\phi_0\right)\right]+\frac{3D\sigma_{12}}{8}\nabla\left(\partial_t\phi\right)\cdot\left(\nabla\phi+\cos\theta\nabla\phi_0\right),
			\end{aligned}	
		\end{equation}
	\end{subequations}
	where the incompressible condition (\ref{eq-NS1}) is applied. With the help of above equations, the variation of free energy can be expressed as
	\begin{equation}\label{eq-DtEf}
		\begin{aligned}
			\frac{\mathrm{d}\mathcal{E}_F}{\mathrm{d}t}=&\frac{\mathrm{d}}{\mathrm{d}t}\int_{\Omega}\left[f_1\left(\phi,\nabla\phi\right)+f_2\left(\phi_0,\phi,\nabla\phi_0,\nabla\phi\right)\right]\mathrm{d}\Omega\\
			=&\int_{\Omega}\left\{\partial_t\phi\left[\frac{3\sigma_{12}}{D}\phi\left(\phi-1\right)\left(\phi+1\right)+\frac{9\sigma_{12}}{D}\phi_0^2\phi+\frac{3\sigma_{12}\cos\theta}{D}\phi_0\left(3\phi^2+\phi_0^2-1\right)\right]+\frac{3D\sigma_{12}}{8}\nabla\left(\partial_t\phi\right)\cdot\left(\nabla\phi+\cos\theta\nabla\phi_0\right)\right\}\mathrm{d}\Omega\\
			=&\int_{\Omega}\left\{\partial_t\phi\mu_{\phi}+\frac{3D\sigma_{12}}{8}\nabla\cdot\left[\partial_t\phi\left(\nabla\phi+\cos\theta\nabla\phi_0\right)\right]\right\}\mathrm{d}\Omega\\
			=&\int_{\partial\Omega}\left[M\mu_{\phi}\nabla\mu_{\phi}+\frac{3D\sigma_{12}}{8}\partial_t\phi\left(\nabla\phi+\cos\theta\nabla\phi_0\right)\right]\cdot\mathrm{d}\mathbf{s}-\int_{\Omega}\mu_{\phi}\mathbf{u}\cdot\nabla\phi\mathrm{d}\Omega-\int_{\Omega}M\nabla\mu_{\phi}\cdot\nabla\mu_{\phi}\mathrm{d}\Omega.
		\end{aligned}
	\end{equation}
	
	Finally, when all boundary integrals vanish with some proper boundary conditions, and the other body forces are ignored except for the pressure gradient, the time derivative of the total energy can be written as
	\begin{equation}
		\frac{\mathrm{d}\mathcal{E}_{total}}{\mathrm{d}t}=\frac{\mathrm{d}\mathcal{E}_F}{\mathrm{d}t}+\frac{\mathrm{d}\mathcal{E}_K}{\mathrm{d}t}=-\int_{\Omega}\left\{M\nabla\mu_{\phi}\cdot\nabla\mu_{\phi}+\frac{\mu}{2\left(1-\phi_0\right)}\left[\nabla\mathbf{u}+\left(\nabla\mathbf{u}\right)^\top\right]:\left[\nabla\mathbf{u}+\left(\nabla\mathbf{u}\right)^\top\right]\right\}\mathrm{d}\Omega\leq0,
	\end{equation}
	where Eqs. (\ref{eq-DtEk}) and (\ref{eq-DtEf}) are adopted. The proof is complete.
\end{proof}

\subsection{Governing equations for free rigid particles}\label{sec-particleEq}
When free rigid particles are considered in the multiphase flows, the following governing equations for particle motion are also needed,
\begin{subequations}\label{eq-particle}
	\begin{equation}
		M_p\frac{\mathrm{d}\mathbf{u}_p}{\mathrm{d}t}=\mathbf{F}_f+\mathbf{F}_f^{in}+\left(1-\frac{\rho}{\rho_p}\right)M_p\mathbf{g},
	\end{equation}
	\begin{equation}
		I_p\frac{\mathrm{d}\bm{\omega}_p}{\mathrm{d}t}=\mathbf{T}_f+\mathbf{T}_f^{in},
	\end{equation}
	\begin{equation}
		\frac{\mathrm{d}\mathbf{x}_p}{\mathrm{d}t}=\mathbf{u}_p,
	\end{equation}
\end{subequations} 
where $M_p$ and $I_p$ are the mass and rotational inertia of the particle, $\rho_p$ is the particle density, $\mathbf{x}_p$ is the mass center of the particle, $\mathbf{u}_p$ is the corresponding translational velocity of the centroid and $\bm{\omega}_p$ is the angular velocity. $\mathbf{g}$ is the gravity acceleration, $\mathbf{F}_f$ and $\mathbf{T}_f$ are the hydrodynamic force and torque, and they are related to the force $\rho\mathbf{f}$,
\begin{equation}\label{eq-Ff}
	\mathbf{F}_f=-\int_{\Omega}\rho\mathbf{f}\mathrm{d}\Omega,\quad \mathbf{T}_f=-\int_{\Omega}\left(\mathbf{x}-\mathbf{x}_p\right)\times\rho\mathbf{f}\mathrm{d}\Omega.
\end{equation}
In addition, as introduced in some IBMs \cite{Suzuki2011CF,Chen2020IJNME}, the force and torque $\mathbf{F}_f^{in}$ and $\mathbf{T}_f^{in}$ are the corrections to the internal fluid effect that compel the velocity of the artificial fluid inside the particle to be consistent with the translational and rotational velocities of the particle \cite{Suzuki2011CF},
\begin{equation}\label{eq-Ffin}
	\mathbf{F}_f^{in}=\frac{\mathrm{d}}{\mathrm{d}t}\int_{V_p}\rho\mathbf{u}\mathrm{d}V_p=\rho V_p\frac{\mathrm{d}\mathbf{u}}{\mathrm{d}t}=\rho V_p\frac{\mathrm{d}\mathbf{u}_p}{\mathrm{d}t},\quad \mathbf{T}_f^{in}=\frac{\mathrm{d}}{\mathrm{d}t}\int_{V_p}I_f\bm{\omega}_f\mathrm{d}V_p=I_f\frac{\mathrm{d}\bm{\omega}_f}{\mathrm{d}t}=I_f\frac{\mathrm{d}\bm{\omega}_p}{\mathrm{d}t},
\end{equation}
where $V_p$ is the volume of the particle, $I_f$ is the moment of inertia for internal fluid, and can be calculated by $I_f=I_p\rho/\rho_p$.

\section{Numerical methods}\label{methods}
In this section, we will present numerical methods for the mathematical models mentioned above. For the governing equations for particle motion in Section \ref{sec-particleEq}, the first-order explicit Euler scheme is applied, and the details are not shown here. In the following, the multiple-relaxation-time (MRT) LB models for phase and flow fields are presented.

The general evolution equation of LB method for convection-diffusion and NS equations can be written as \cite{Chai2020PRE,Chai2023PRE}
\begin{equation}\label{eq-LBE}
	f_{i,k}\left(\mathbf{x}+\mathbf{c}_i\Delta t,t+\Delta t\right)=f_{i,k}\left(\mathbf{x},t\right)-\Lambda_{ij}^k\left[f_{j,k}\left(\mathbf{x},t\right)-f_{j,k}^{eq}\left(\mathbf{x},t\right)\right]+\Delta t\left(\delta_{ij}-\frac{\Lambda_{ij}}{2}\right)F_{j,k},
\end{equation} 
where $f_{i,k}\left(\mathbf{x},t\right)$ is the distribution function at position $\mathbf{x}$ and time $t$ along the $i$-th direction in the discrete velocity space ($k=1,2,\cdots$ for different fields), $f_{i,k}^{eq}\left(\mathbf{x},t\right)$ is the corresponding equilibrium distribution function, and $F_{i,k}\left(\mathbf{x},t\right)$ is the distribution function of source/force term. $\mathbf{c}_i$ is the lattice speed, $\Delta t$ is the time step, and $\bm{\Lambda}^k=\left(\Lambda_{ij}^k\right)$ represents an invertible collision matrix. 

It is worth noting that through properly choosing some specific moments of the distribution functions $f_{i,k}^{eq}$ and $F_{i,k}$, and the relation between the characteristic parameter in the governing equation and an eigenvalue of the collision matrix, one can get the LB models for different macroscopic equations via some asymptotic analysis methods, such as the direct Taylor expansion or Chapman-Enskog analysis.

\subsection{Lattice Boltzmann model for phase field}
Similar to the LB model for CH equation in some previous works \cite{Liang2014PRE,Wang2019Capillarity}, the equilibrium and source distribution functions are expressed as
\begin{equation}
	f_{i,1}^{eq}=\begin{cases}
		\phi+\left(\omega_i-1\right)\eta\mu_{\phi},\quad &i=0\\
		\omega_i\eta\mu_{\phi}+\omega_i\mathbf{c}_i\cdot\phi\mathbf{u}/c_s^2,\quad &i\neq0
	\end{cases},\quad
	F_{i,1}=\frac{\omega_i\mathbf{c}_i\cdot\partial_t\left(\phi\mathbf{u}\right)}{c_s^2},
\end{equation}
where $\omega_i$ is the weight coefficient, $c_s$ is the lattice sound speed, and $\eta$ represents an adjustable parameter.

Through the direct Taylor expansion under the general LB framework \cite{Chai2020PRE,Chai2023PRE}, Eq. (\ref{eq-CHE}) can be correctly recovered from the LB evolution equation with the relation $M=\left(1/s_{1\phi}-1/2\right)\eta c_s^2\Delta t$, where $s_{1\phi}$ is an eigenvalue of the collision matrix $\bm{\Lambda}$.
Finally, the order parameter is computed by
\begin{equation}\label{eq-phi}
	\phi=\sum_if_{i,1}.
\end{equation}
\subsection{Lattice Boltzmann model for flow field}
To derive the incompressible NS equations, the equilibrium distribution function for the flow field is designed as
\begin{equation}
	f_{i,2}^{eq}=\lambda_i+\frac{\omega_i\mathbf{c}_i\cdot\rho\mathbf{u}}{c_s^2}+\omega_i\left[\frac{\left(c_{i\alpha}c_{i\alpha}-c_s^2\right)\left(\rho u_{\alpha}u_{\alpha}-S_{\alpha}u_{\alpha}\right)}{c^2c_s^2-c_s^4}+\frac{c_{i\alpha}c_{i\bar{\alpha}}}{c_s^4}\left(\rho u_{\alpha}u_{\bar{\alpha}}-\frac{S_{\alpha}u_{\bar{\alpha}}+u_{\alpha}S_{\bar{\alpha}}}{2}\right)\right],
\end{equation}
where $\lambda_0=\left(\omega_0-1\right)P/c_s^2$, $\lambda_i=\omega_iP/c_s^2$ $\left(i\neq0\right)$, $\alpha=1,2,\cdots,d$, $\alpha<\bar{\alpha}\leq d$ with $d$ being the dimensionality. $c=\Delta x/\Delta t$, and $\Delta x$ is the lattice spacing. Here the two terms in the square bracket correspond to the diagonal and non-diagonal parts of the momentum flux, respectively. In addition, the distribution function of force term can be given by
\begin{equation}
	F_{i,2}=\omega_i\left[\mathbf{u}\cdot\nabla\rho+\frac{\mathbf{c}_i\cdot\left(\mathbf{F}+\rho\mathbf{f}\right)}{c_s^2}+\frac{\left(c_{i\alpha}c_{i\alpha}-c_s^2\right)M_{\alpha\alpha}^{2F}}{c^2c_s^2-c_s^4}+\frac{c_{i\alpha}c_{i\bar{\alpha}}M_{\alpha\bar{\alpha}}^{2F}}{c_s^4}\right],
\end{equation}
where $\mathbf{F}=\mu_{\phi}\nabla\phi+\mu_{\phi_0}\nabla\phi_0+\mathbf{F}_b+\mathbf{F}_c$, $\mathbf{F}_c=\nabla\cdot\left(\mathbf{Su}-\mathbf{uS}\right)/2$ is a corrected force term to recover the consistent momentum equation (\ref{eq-NS2}), and the tensor $\mathbf{M}_{2F}$ is defined by
\begin{equation}
	\mathbf{M}_{2F}=\partial_t\left(\rho\mathbf{u}\mathbf{u}-\frac{\mathbf{S}\mathbf{u}+\mathbf{u}\mathbf{S}}{2}\right)+c_s^2\mathbf{u}\nabla\rho+c_s^2\left(\nabla\rho\right)\mathbf{u}+\left(c^2-3c_s^2\right)\left(\mathbf{u}\cdot\nabla\rho\right)\mathbf{I}.
\end{equation}

Based on above expressions, the NS equations (\ref{eq-NSE}) can be correctly recovered with the following relation \cite{Chai2023PRE}, 
\begin{equation}
	\frac{\mu}{\rho\left(1-\phi_0\right)}=\left(\frac{1}{s_{2a}}-\frac{1}{2}\right)\frac{c^2-c_s^2}{2}\Delta t=\left(\frac{1}{s_{2b}}-\frac{1}{2}\right)c_s^2\Delta t,
\end{equation} 
where $s_{2a}$ and $s_{2b}$ are two relaxation parameters, and will be stated below.

Finally, the macroscopic velocity and pressure are calculated by
\begin{subequations}
	\begin{equation}\label{eq-uStar}
		\mathbf{u}^*=\frac{1}{\rho}\sum_i\mathbf{c}_if_{i,2}+\frac{\Delta t}{2\rho}\mathbf{F}, 
	\end{equation}
	\begin{equation}\label{eq-u}
		\mathbf{u}=\mathbf{u}^*+\frac{\Delta t}{2}\mathbf{f},\quad P=\frac{c_s^2}{1-\omega_0}\left[\sum_{i\neq0}f_{i,2}+\left(\frac{1}{2}+H\right)\Delta t\mathbf{u}\cdot\nabla\rho-\omega_0\frac{\left(\rho\mathbf{u}-\mathbf{S}\right)\cdot\mathbf{u}}{c^2-c_s^2}+\frac{K}{c^2}\frac{\Delta t}{2}\partial_t\left(\rho\mathbf{u}\cdot\mathbf{u}-\mathbf{S}\cdot\mathbf{u}\right)\right],
	\end{equation}
\end{subequations}
where $\mathbf{u}^*$ is the velocity without considering the fluid-solid interaction, $\mathbf{u}$ is the corrected velocity. The force $\mathbf{f}$ can be discretized by $\phi_0\left(\mathbf{u}_s-\mathbf{u}^*\right) /\Delta t$ with $\mathbf{u}_s=\mathbf{u}_p+\bm{\omega}_p\times\left(\mathbf{x}-\mathbf{x}_p\right)$ being the velocity of the solid point. $K$ and $H$ are two factors, for instance, $K=1-d_0$ and $H=K\left[\left(1-d_0\right)/s_0-\left(2-4d_0\right)/s_{2a}+\left(1-3d_0\right)/2\right]$ for two-dimensional problem, while for three-dimensional case, $K=1-2d_0$ and $H=K\left[\left(1-d_0\right)/s_0-\left(3-7d_0\right)/s_{2a}+\left(1-3d_0\right)\right]$. 
Here $s_0$ is a relaxation parameter related to the collision matrix, $d_0=c_s^2/c^2$ is an adjustable scale factor, and is fixed as $1/3$ in classical MRT-LB methods. $\partial_t\left(\rho\mathbf{uu}\right)$ is usually approximated by $\mathbf{u}\mathbf{F}+\mathbf{F}\mathbf{u}$, while other temporal derivatives are discretized by the explicit Euler scheme, and the gradient and Laplacian operators are computed by the second-order isotropic central schemes \cite{Zhan2022PRE,Wang2019Capillarity}.

\subsection{Computational procedure}
The computational produce of the numerical algorithm is listed as follows.
\begin{itemize}
	\item[1.] Initialize physical variables and the distribution functions by the equilibrium distribution functions.
	\item[2.] Implement LB evolution equation (\ref{eq-LBE}) for phase and flow fields in the computational domain.
	\item[3.] Compute the order parameter $\phi$ and density $\rho$ through Eqs. (\ref{eq-phi}) and (\ref{eq-rho}).
	\item[4.] Calculate the surface tension forces $\mu_{\phi}\nabla\phi$ and $\mu_{\phi_0}\nabla\phi_0$.
	\item[5.] Compute the velocity $\mathbf{u}^*$ and the fluid-solid interaction force $\mathbf{f}$.
	\item[6.] Update the velocity and pressure by using Eq. (\ref{eq-u}).
	\item[7.] Calculate the hydrodynamic force $\mathbf{F}_f$ and torque $\mathbf{T}_f$ through the discretization of Eq. (\ref{eq-Ff}).
	\item[8.] Compute the translational velocity, angular velocity, and position of the particle according to Eq. (\ref{eq-particle}).
	\item[9.] Update the smooth indicator function $\phi_0$ for the particle by using the hyperbolic tangent function.
	\item[10.] Calculate the inside force $\mathbf{F}_f^{in}$ and torque $\mathbf{T}_f^{in}$ according to the explicit discretization of Eq. (\ref{eq-Ffin}).
	\item[11.] Repeat steps 2–10 until the convergence is reached.
\end{itemize}

\section{Numerical validations and discussion}\label{Simulations}
In the framework of LB method, there are some LB models based on different collision matrices, such as the classical MRT-LB model \cite{Humieres1992,Lallemand2000PRE,Coveney2002}, the cascaded or central-moment LB model \cite{Geier2006PRE,Premnath2012CiCP}, the Hermite moment LB model \cite{Coreixas2017PRE}, central-Hermite moment LB model \cite{Coreixas2019PRE,Mattila2017PF} and so on. However, under the unified framework of modeling of the MRT-LB model in Refs. \cite{Chai2020PRE,Chai2023PRE}, different forms of the collision matrices can be converted into each other through a specific lower triangular matrix with the diagonal element of unity. 
In this work, the Hermite moment model is applied, and the collision matrix can be written as
\begin{equation}
	\bm{\Lambda}^k=\mathbf{M}_{H}^{-1}\mathbf{S}_{H}^k\mathbf{M}_{H}=\mathbf{M}_0^{-1}\mathbf{S}_0^k\mathbf{M}_0,
\end{equation} 
where $\mathbf{M}_{H}$ and $\mathbf{S}_{H}^k$ are the transformation matrix and relaxation matrix of the Hermite-moment based LB model. As shown in Ref. \cite{Chai2023PRE}, $\mathbf{M}_{H}$ can be written as $\mathbf{N}_{H}\mathbf{M}$, $\mathbf{N}_{H}$ is the lower triangular matrix associated with $\mathbf{M}_{H}$ and natural transformation matrix $\mathbf{M}=\mathbf{C}_d\mathbf{M}_0$, and $\mathbf{C}_d$ is a diagonal matrix formed by the powers of lattice speed $c$. $\mathbf{S}_0^k=\mathbf{C}_d^{-1}\mathbf{N}_{H}^{-1}\mathbf{S}_{H}^k\mathbf{N}_{H}\mathbf{C}_d$ is the new lower triangular relaxation matrix in the moment space formed by $\mathbf{M}_0$. In our simulations for two-dimensional problems, the D2Q9 lattice structure is applied to both phase field and flow field, and the weight coefficients in all directions are given by $\omega_0=1-2d_0+d_0^2$, $\omega_{1-4}=\left(d_0-d_0^2\right)/2$, and $\omega_{5-8}=d_0^2/4$. Other matrices are listed as follows:
\begin{equation}
	\begin{aligned}
		\mathbf{C}_d&=\mathbf{diag}\left(1,c,c,c^2,c^2,c^2,c^3,c^3,c^4\right),\\
		\mathbf{S}_{H}^1&=\mathbf{diag}\left(s_0,s_{1\phi},s_{1\phi},s_2,s_2,s_2,s_3,s_3,s_4\right),\quad\mathbf{S}_{H}^2=\mathbf{diag}\left(s_0,s_1,s_1,s_{2a},s_{2a},s_{2b},s_3,s_3,s_4\right),\\
		\mathbf{M}_0&=\begin{pmatrix*}[r]
			1 &  1 &  1 &  1 &  1 &  1 &  1 &  1 &  1\\
			0 &  1 &  0 & -1 &  0 &  1 & -1 & -1 &  1\\
			0 &  0 &  1 &  0 & -1 &  1 &  1 & -1 & -1\\
			0 &  1 &  0 &  1 &  0 &  1 &  1 &  1 &  1\\
			0 &  0 &  1 &  0 &  1 &  1 &  1 &  1 &  1\\		
			0 &  0 &  0 &  0 &  0 &  1 & -1 &  1 & -1\\	
			0 &  0 &  0 &  0 &  0 &  1 & -1 & -1 &  1\\	
			0 &  0 &  0 &  0 &  0 &  1 &  1 & -1 & -1\\
			0 &  0 &  0 &  0 &  0 &  1 &  1 &  1 &  1\\
		\end{pmatrix*},\quad
		\mathbf{N}_{H}=\begin{pmatrix*}[r]
			1 & 0 & 0 & 0 & 0 & 0 & 0 & 0 & 0 \\
			0 & 1 & 0 & 0 & 0 & 0 & 0 & 0 & 0 \\
			0 & 0 & 1 & 0 & 0 & 0 & 0 & 0 & 0 \\
			-c_s^2 & 0 & 0 & 1 & 0 & 0 & 0 & 0 & 0 \\
			-c_s^2 & 0 & 0 & 0 & 1 & 0 & 0 & 0 & 0 \\
			0 & 0 & 0 & 0 & 0 & 1 & 0 & 0 & 0 \\
			0 & -c_s^2 & 0 & 0 & 0 & 0 & 1 & 0 & 0 \\
			0 & 0 & -c_s^2 & 0 & 0 & 0 & 0 & 1 & 0 \\
			c_s^4 & 0 & 0 & -c_s^2 & -c_s^2 & 0 & 0 & 0 & 1 \\
		\end{pmatrix*}.
	\end{aligned}
\end{equation}

In the following, some problems are used to test the present diffuse-interface model for the gas-liquid-solid multiphase flows.

\begin{figure}
	\centering
	\includegraphics[width=1.5in]{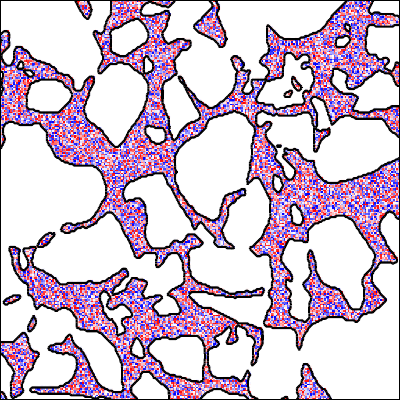}
	\caption{The random initialization of two phases in a porous medium.}
	\label{fig-porousInit}
\end{figure}

\subsection{Phase separation in a porous medium}\label{sec-porous}
We first consider the phase separation in a porous medium based on a $256\times256$ micro-CT image of the poorly sorted unconsolidated fluvial sandpack \cite{Santos2022DB}. The no-flux boundary condition is imposed on all boundaries, and some parameters are set as $M=0.1$, $\sigma_{12}=0.001$, $D=4$, $\rho_1=\rho_2=1$, $\mu_1=\mu_2=0.1$, $d_0=0.4$, $s_i=1$, and $\Delta x=\Delta t=1$. To adopt the present diffuse-interface method, the original data of this porous medium needs to be smoothed by a standard CH equation with finite-time iterations (20 times in this work). As shown in Fig. \ref{fig-porousInit}, the phase-field variable is randomly initialized by $\phi\left(x,y\right)=\left[1-\phi_0\left(x,y\right)\right]\text{rand}(x,y)$, where $\text{rand}(x,y)$ is a function with the random value between -1 and 1. We perform some simulations at the contact angles of $\theta=45^\circ,90^\circ,135^\circ$, and plot the evolution process in Fig. \ref{fig-porous}. In this figure, the droplet sizes increase in time, and some of them also coalesce into the larger ones, which leads to the eventual separation of binary fluid components. Due to the difference of the wetting property on the solid surface, the droplets also show hydrophilic or hydrophobic phenomena in the long-time evolutions, and the accuracy of present diffuse-interface model in predicting the contact angle will be tested in detail in the next part.
Figure \ref{fig-porousEnergy} shows the evolutions of the total energy $\mathcal{E}_{toal}$ at different contact angles. From this figure, one can observe that the total energy of this system decreases with time, which is consistent with the energy dissipation law given by Theorem \ref{Theorem}.
\begin{figure}
	\centering
	\subfigure[$\theta=45^\circ$]{
		\begin{minipage}{0.24\linewidth}
			\centering
			\includegraphics[width=1.5in]{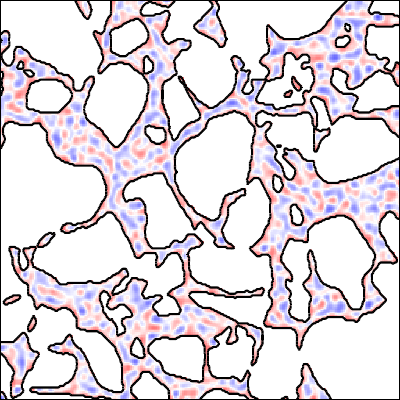}
		\end{minipage}
		\begin{minipage}{0.24\linewidth}
			\centering
			\includegraphics[width=1.5in]{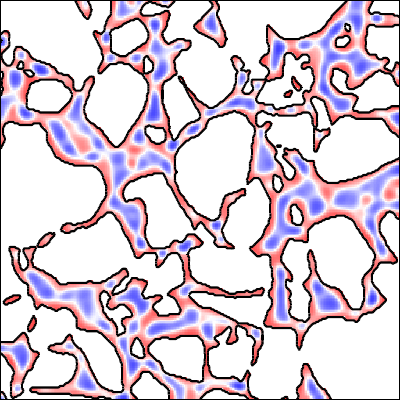}
		\end{minipage}
		\begin{minipage}{0.24\linewidth}
			\centering
			\includegraphics[width=1.5in]{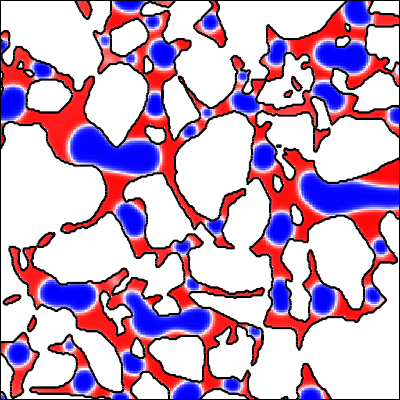}
		\end{minipage}
		\begin{minipage}{0.24\linewidth}
			\centering
			\includegraphics[width=1.5in]{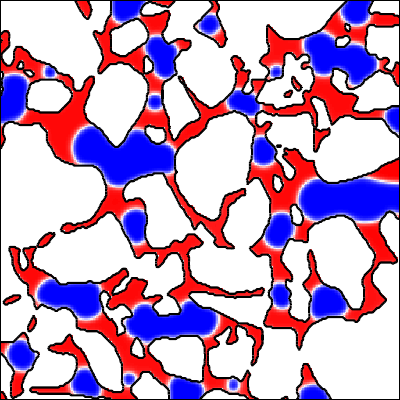}
	\end{minipage}}
	
	\subfigure[$\theta=90^\circ$]{
		\begin{minipage}{0.24\linewidth}
			\centering
			\includegraphics[width=1.5in]{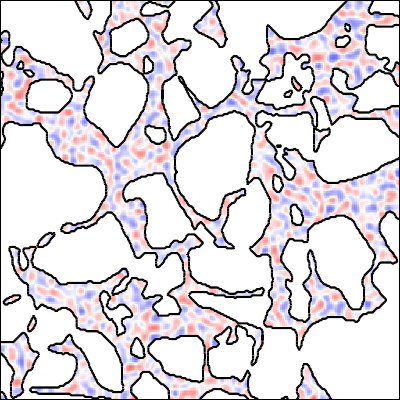}
		\end{minipage}
		\begin{minipage}{0.24\linewidth}
			\centering
			\includegraphics[width=1.5in]{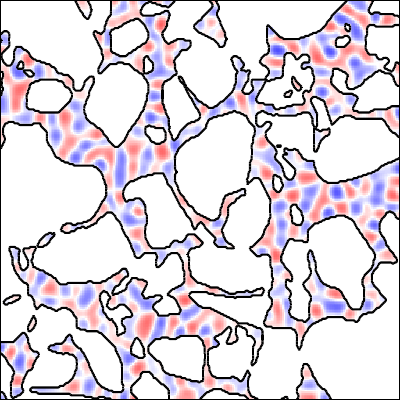}
		\end{minipage}
		\begin{minipage}{0.24\linewidth}
			\centering
			\includegraphics[width=1.5in]{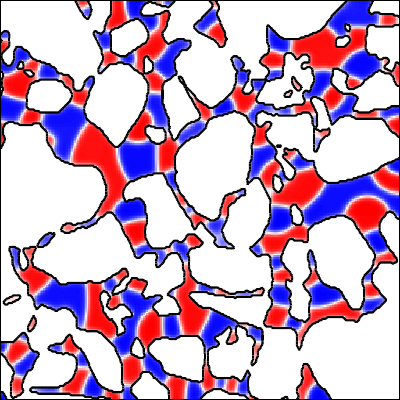}
		\end{minipage}
		\begin{minipage}{0.24\linewidth}
			\centering
			\includegraphics[width=1.5in]{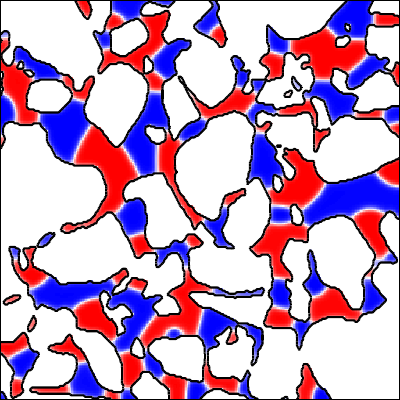}
	\end{minipage}}
	
	\subfigure[$\theta=135^\circ$]{
		\begin{minipage}{0.24\linewidth}
			\centering
			\includegraphics[width=1.5in]{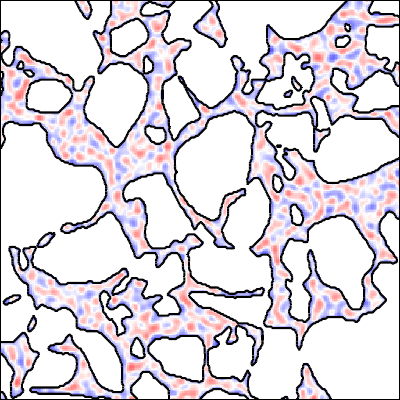}
		\end{minipage}
		\begin{minipage}{0.24\linewidth}
			\centering
			\includegraphics[width=1.5in]{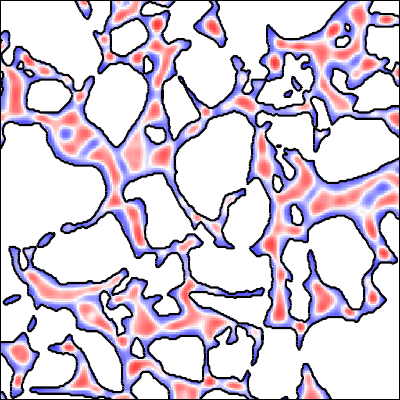}
		\end{minipage}
		\begin{minipage}{0.24\linewidth}
			\centering
			\includegraphics[width=1.5in]{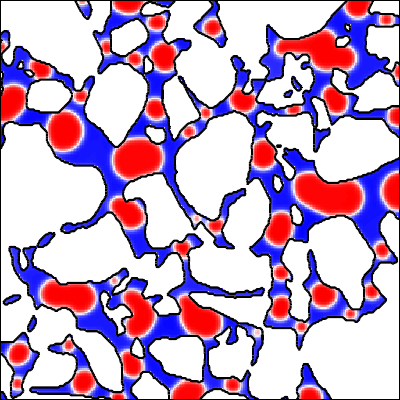}
		\end{minipage}
		\begin{minipage}{0.24\linewidth}
			\centering
			\includegraphics[width=1.5in]{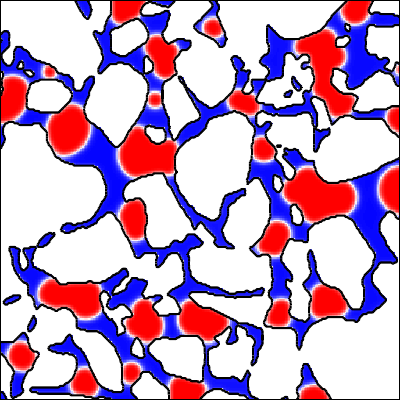}
	\end{minipage}}
	\caption{The evolution processes of the binary phase separation in the porous medium at $t=1\times10^4,1\times10^5,1\times10^6,1\times10^7$ (from left to right columns).}
	\label{fig-porous}
\end{figure}
\begin{figure}
	\centering
	\includegraphics[width=3.5in]{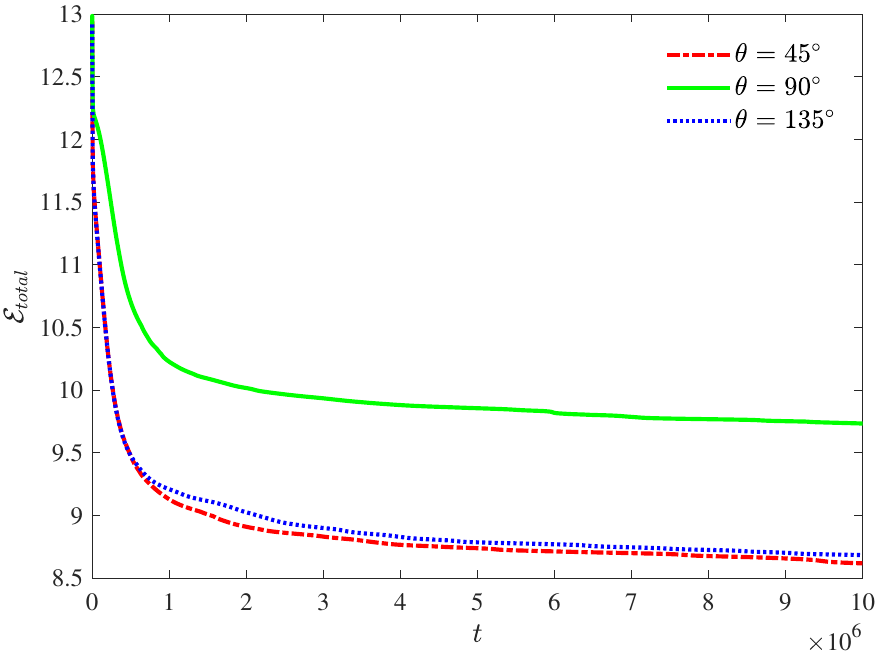}
	\caption{The evolution of the total energy $\mathcal{E}_{total}$ for the binary phase separation in a porous medium.}
	\label{fig-porousEnergy}
\end{figure}
\subsection{A droplet spreading on a stationary cylinder surface}\label{sec-cylinder}
\begin{figure}
	\centering
	\subfigure[$\theta=30^\circ$]{
		\begin{minipage}{0.24\linewidth}
			\centering
			\includegraphics[width=1.5in]{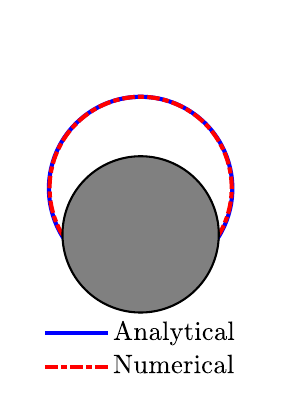}
	\end{minipage}}
	\subfigure[$\theta=60^\circ$]{
		\begin{minipage}{0.24\linewidth}
			\centering
			\includegraphics[width=1.5in]{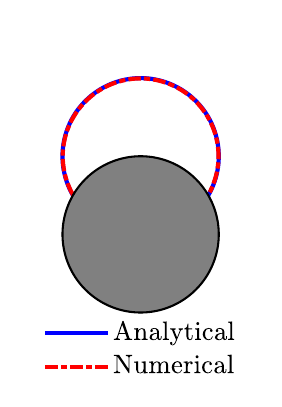}
	\end{minipage}}	
	\subfigure[$\theta=90^\circ$]{
		\begin{minipage}{0.24\linewidth}
			\centering
			\includegraphics[width=1.5in]{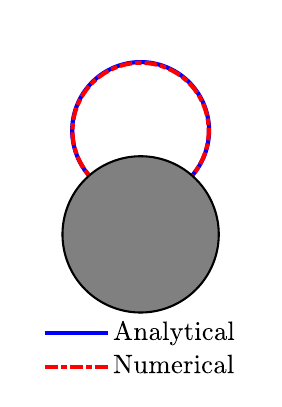}
	\end{minipage}}
	\subfigure[$\theta=120^\circ$]{
		\begin{minipage}{0.24\linewidth}
			\centering
			\includegraphics[width=1.5in]{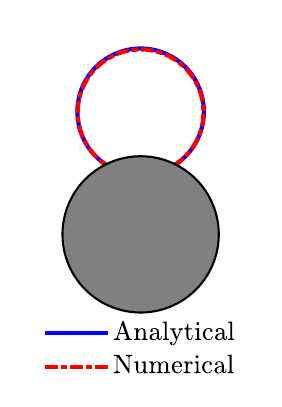}
	\end{minipage}}
	\caption{The equilibrium shapes of the droplet at different contact angles.}
	\label{fig-cylinder}
\end{figure}
\begin{figure}
	\centering
	\includegraphics[width=3.5in]{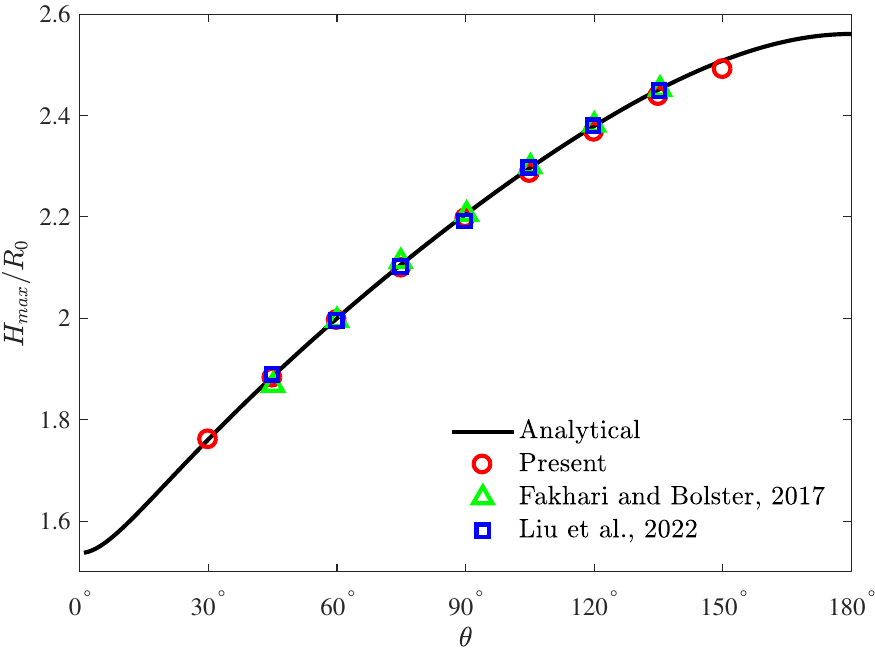}
	\caption{The normalized droplet heights at different contact angles.}
	\label{fig-cylinderH}
\end{figure}
\begin{figure}
	\centering
	\includegraphics[width=3.5in]{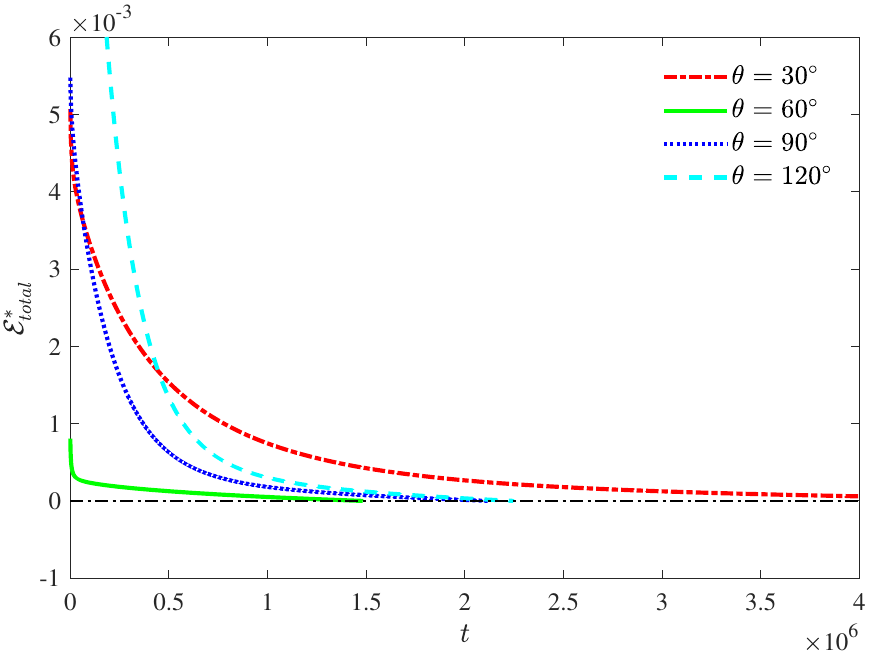}
	\caption{The evolution of normalized total energy $\mathcal{E}^*_{total}=\mathcal{E}_{total}/\mathcal{E}^{eq}_{total}-1$, $\mathcal{E}^{eq}_{total}$ is the total energy at the final equilibrium moment.}
	\label{fig-cylinderEnergy}
\end{figure}

In this part, we continue to consider a droplet spreading on a cylinder surface in the square domain $[-L/2,L/2]\times[-L/4,3L/4]$ with $L=256$ to test the capacity of the present diffuse-interface model in capturing the fluid interface and predicting the contact angle. In the following simulations, the periodic boundary condition is used in $x$-direction and no-flux boundary condition is imposed at the bottom and top boundaries, a solid cylinder with radius $R_p$ is located at $\left(0,0\right)$, and a droplet with radius $R_0$ is placed on the cylinder surface. Some physical parameters are set as $R_0=R_p=50$, $\rho_1=1$, $\rho_2=0.1$, and the others are the same as those in Section \ref{sec-porous}. Initially, the phase-field variables are given as
\begin{equation}
	\begin{aligned}
		&\phi_0\left(x,y\right)=\frac{1}{2}+\frac{1}{2}\tanh\frac{R_p-\sqrt{x^2+y^2}}{D/2},\\
		&\phi\left(x,y\right)=\left[1-\phi_0\left(x,y\right)\right]\tanh\frac{R_0-\sqrt{x^2+\left(y-R_p\right)^2}}{D/2}.
	\end{aligned}
\end{equation}

We conduct some simulations and plot the final equilibrium shapes of the droplet at a large range of contact angles in Fig. \ref{fig-cylinder}. As shown in this figure, the shapes of the droplet agree well with the analytical solutions. To give a quantitative comparison, the maximum height of the droplet at equilibrium state ($H_{max}$) is measured, and it actually can be obtained by the initial shape under the condition of equal area,
\begin{equation}
	H_{max}=r+k,\quad \pi r^2-r^2\cos^{-1}\left(\frac{k^2+r^2-R_p^2}{2kr}\right)-R_p^2\cos^{-1}\left(\frac{k^2-r^2+R_p^2}{2kR_p}\right)+k\sqrt{R_p^2-\left(\frac{k^2-r^2+R_p^2}{2k}\right)}=A_0,
\end{equation}
where $A_0$ is the initial area, $k$ is related to the finial radius $r$ and contact angle $\theta$,
\begin{equation}
	k=\sqrt{r^2+R_p^2-2rR_p\cos\theta}.
\end{equation}
Figure \ref{fig-cylinderH} presents the corresponding normalized maximum height of the droplet, and the present results are in a good agreement with the analytical and previous data \cite{Fakhari2017JCP,Liu2022MMS}. Additionally, there is no external force in this system such that the total energy should also decrease with time, and the numerical results in Fig. \ref{fig-cylinderEnergy} indeed confirm this property.

\subsection{Sedimentation of a particle}
\begin{figure}
	\centering
	\subfigure[]{
		\begin{minipage}{0.48\linewidth}
			\centering
			\includegraphics[width=3.0in]{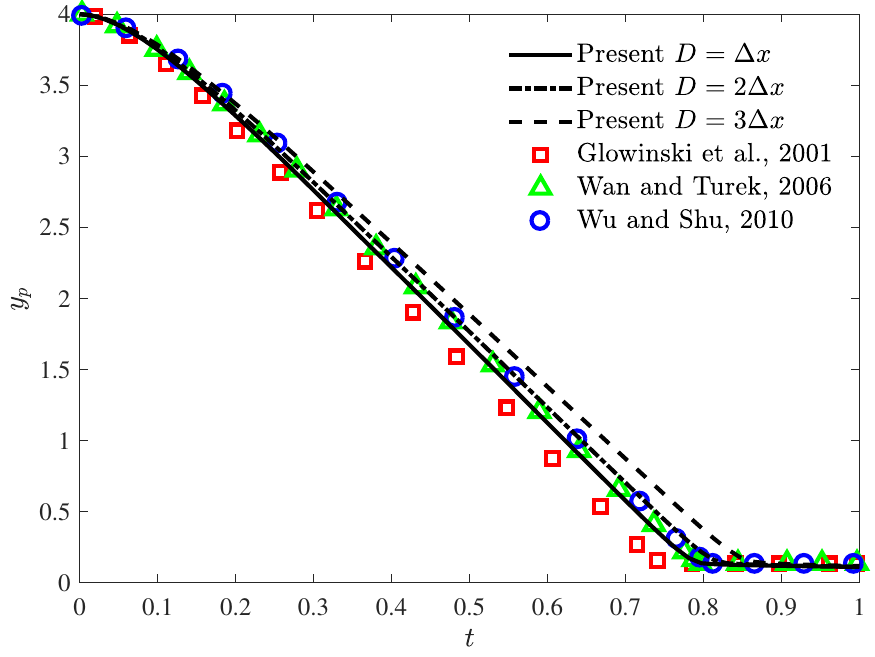}
	\end{minipage}}
	\subfigure[]{
		\begin{minipage}{0.48\linewidth}
			\centering
			\includegraphics[width=3.0in]{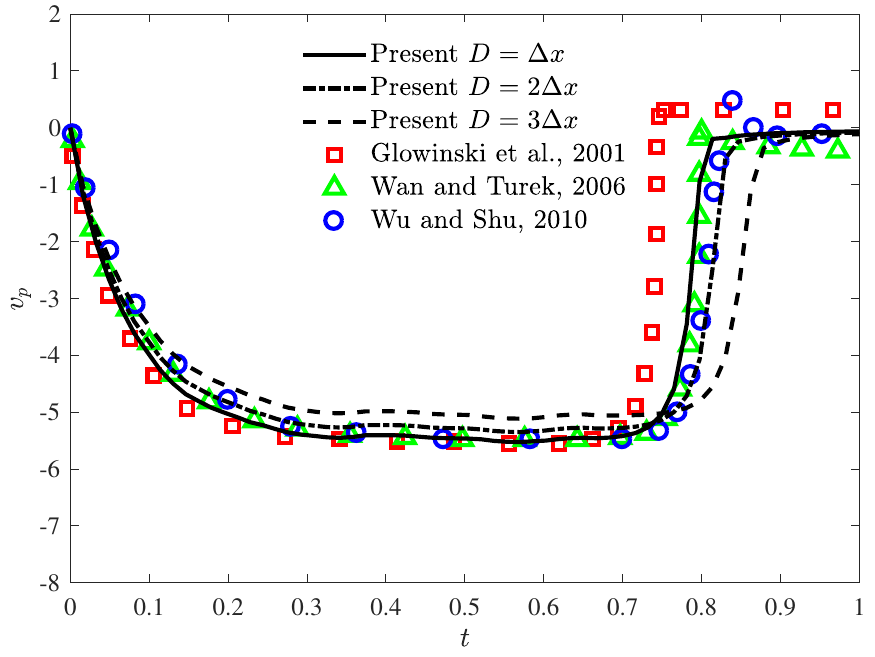}
	\end{minipage}}
	\caption{Time histories of some quantities of the particle [(a) the vertical position $y_p$, (b) the vertical translational velocity $v_p$].}
	\label{fig-impact0}
\end{figure}

We further consider the problem of a single particle sedimentation in a stratified medium to test the accuracy of the present diffuse-interface model in predicting the hydrodynamic force. Initially, a circular particle with the radius $R_p=0.125$ and density $\rho_p=1.25$ is placed at the centerline of a channel with the width $W=2$ and length $H=6$, and it will free fall from the position $\left(1,4\right)$ under the gravity acceleration $|\mathbf{g}|=980$. In the channel, the fluid 1 is placed in the lower region of $y\leq H/2$ and the upper half region $y>H/2$ is filled with fluid 2, and the no-slip boundary condition is applied on all boundaries. Additionally, a body force $\left(\rho-\rho_2\right)\mathbf{g}$ is imposed on the fluids. 
In the simulation of this problem, the lattice spacing and time step are fixed as $\Delta x=1/192$ and $\Delta t=\Delta x/64$. We first focus on the homogeneous medium with $\rho_1=\rho_2=1$ and $\mu_1=\mu_2=0.1$, and present the time histories of the vertical position $y_p$ and the vertical translational velocity $v_p$ of the particle in Fig. \ref{fig-impact0}, from which one can find that when the interface width $D$ becomes smaller, the numerical results are more consistent with the data based on some sharp-interface methods \cite{Glowinski2001JCP,Wan2006IJNMF,Wu2010CiCP}. 

However, when we consider the stratified medium with a density difference by the phase-field model, the interface width should be greater than $2\Delta x$ at least to ensure to capture the large deformation of the fluid interface. We then present a comparison of different values of $\rho_1$ in Fig. \ref{fig-impact} where the interface width is set to be $D=2\Delta x$ and other parameters are the same as those stated previously. As shown in Fig. \ref{fig-impact}(a), when $\rho_1=1.1$, the particle needs more time to fall to the bottom boundary, compared to the homogeneous case, this is because it is subjected to a larger buoyancy in fluid 1. When $\rho_1=1.2$, the buoyancy becomes much larger but still less than the gravity of the particle, and thus the particle finally falls to the bottom of the channel after some oscillations at the fluid interface and the long-time sedimentation in fluid 1. For the last case with $\rho_1=1.3$, the gravity of the particle is smaller than the buoyancy from fluid 1 such that the particle oscillates at the fluid interface until the kinetic energy is dissipated and reaches a steady state. What is more, the wettability almost has no effect on the vertical position of the particle, except for the case of $\rho_1=1.2$ where the sedimentation is slower with a larger contact angle. The time histories of the vertical velocity of a neutral particle with different values of the density $\rho_1$ are also shown in Fig. \ref{fig-impact}(b), from which one can observe that the velocity takes a longer time to go back to zero at a larger density $\rho_1$.     
In addition, we also present the dynamics of the particle and the surrounding fluids in Fig. \ref{fig-impact125}, where $\rho_1=1.2$, $\theta=45^\circ$ and $135^\circ$. From this figure, one can see that a hydrophilic particle immerses in fluid 1 faster than a hydrophobic particle. When the particle is immersed in fluid 1, the hydrophilic particle is fully surrounded by fluid 1, while a bubble is adhered to the hydrophobic particle all the way. 

\begin{figure}
	\centering
	\subfigure[]{
		\begin{minipage}{0.48\linewidth}
			\centering
			\includegraphics[width=3.0in]{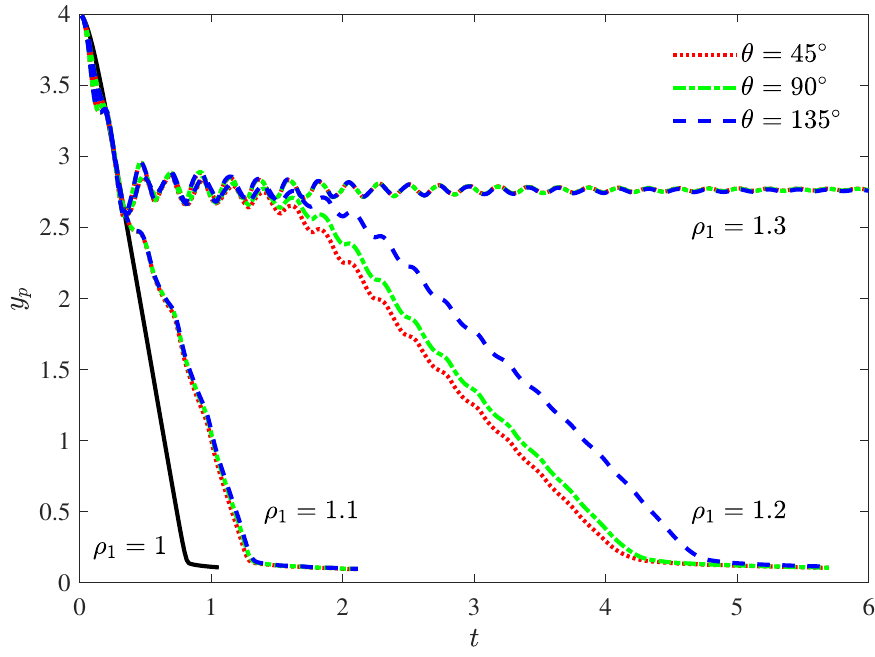}
	\end{minipage}}
	\subfigure[]{
		\begin{minipage}{0.48\linewidth}
			\centering
			\includegraphics[width=3.0in]{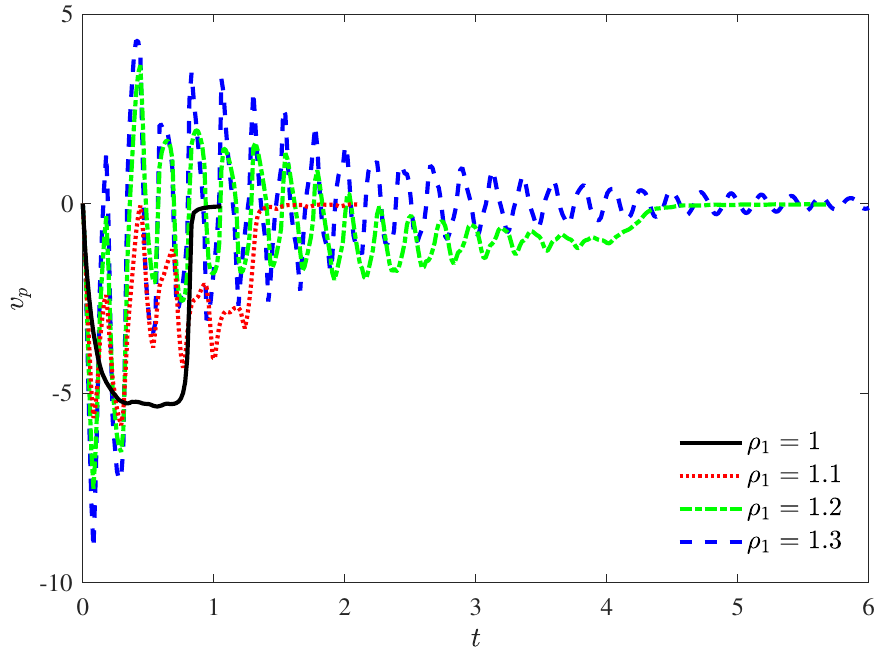}
	\end{minipage}}
	\caption{Time histories of some quantities of the particle [(a) the vertical position $y_p$, (b) the vertical translational velocity $v_p$].}
	\label{fig-impact}
\end{figure}
\begin{figure}
	\centering
	\subfigure[$\theta=45^\circ$]{
		\begin{minipage}{0.15\linewidth}
			\centering
			\includegraphics[width=0.9in]{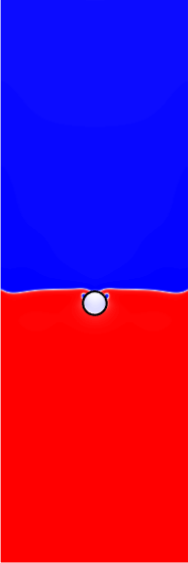}
	\end{minipage}
		\begin{minipage}{0.15\linewidth}
			\centering
			\includegraphics[width=0.9in]{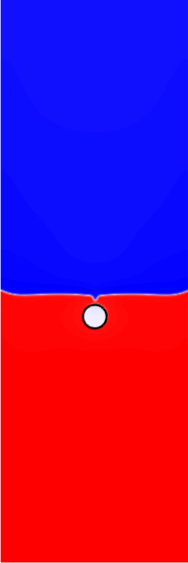}
	\end{minipage}
		\begin{minipage}{0.15\linewidth}
			\centering
			\includegraphics[width=0.9in]{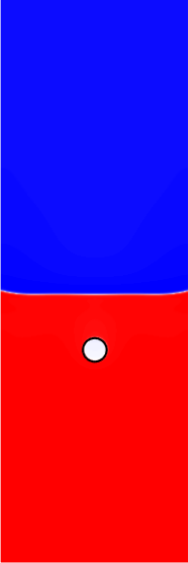}
	\end{minipage}
		\begin{minipage}{0.15\linewidth}
			\centering
			\includegraphics[width=0.9in]{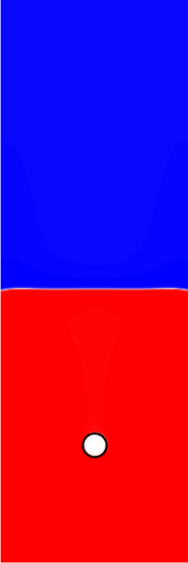}
	\end{minipage}
		\begin{minipage}{0.15\linewidth}
			\centering
			\includegraphics[width=0.9in]{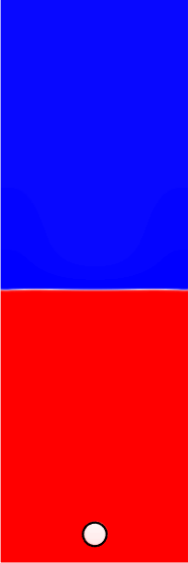}
	\end{minipage}}
	\subfigure[$\theta=135^\circ$]{
		\begin{minipage}{0.15\linewidth}
			\centering
			\includegraphics[width=0.9in]{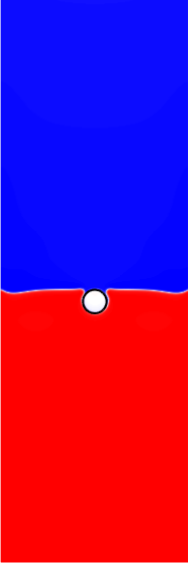}
		\end{minipage}
		\begin{minipage}{0.15\linewidth}
			\centering
			\includegraphics[width=0.9in]{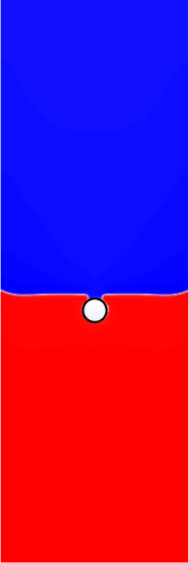}
		\end{minipage}
		\begin{minipage}{0.15\linewidth}
			\centering
			\includegraphics[width=0.9in]{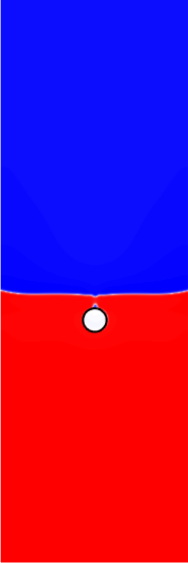}
		\end{minipage}
		\begin{minipage}{0.15\linewidth}
			\centering
			\includegraphics[width=0.9in]{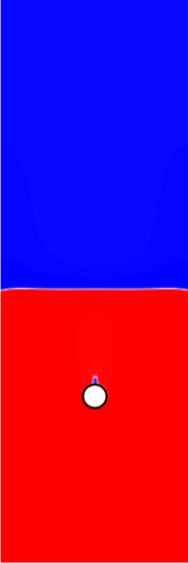}
		\end{minipage}
		\begin{minipage}{0.15\linewidth}
			\centering
			\includegraphics[width=0.9in]{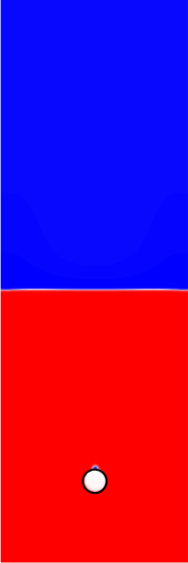}
	\end{minipage}}
	\caption{Snapshots of the particle sedimentation in a stratified medium at $t=1,1.5,2,3,4$ (from left to right columns).}
	\label{fig-impact125}
\end{figure}

\subsection{Balance of wetting particles on the liquid-liquid interface}
\begin{figure}
	\centering
	\subfigure[$Bo=-0.4096$]{
		\begin{minipage}{0.32\linewidth}
			\centering
			\includegraphics[width=1.5in]{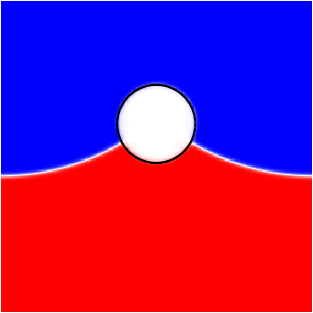}
	\end{minipage}}
	\subfigure[$Bo=0$]{
		\begin{minipage}{0.32\linewidth}
			\centering
			\includegraphics[width=1.5in]{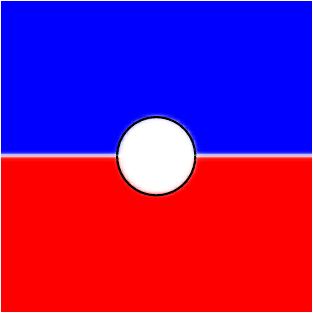}
	\end{minipage}}
	\subfigure[$Bo=0.4096$]{
		\begin{minipage}{0.32\linewidth}
			\centering
			\includegraphics[width=1.5in]{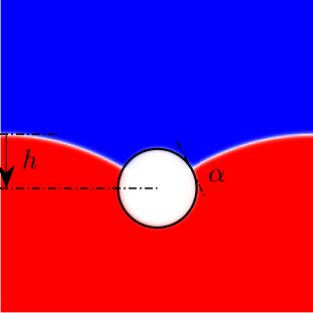}
	\end{minipage}}
	\caption{Snapshots of the mechanical equilibrium states of a neutral particle trapped at fluid-fluid interface in presence of gravity.}
	\label{fig-1particle}
\end{figure}
\begin{figure}
	\centering
	\includegraphics[width=3.5in]{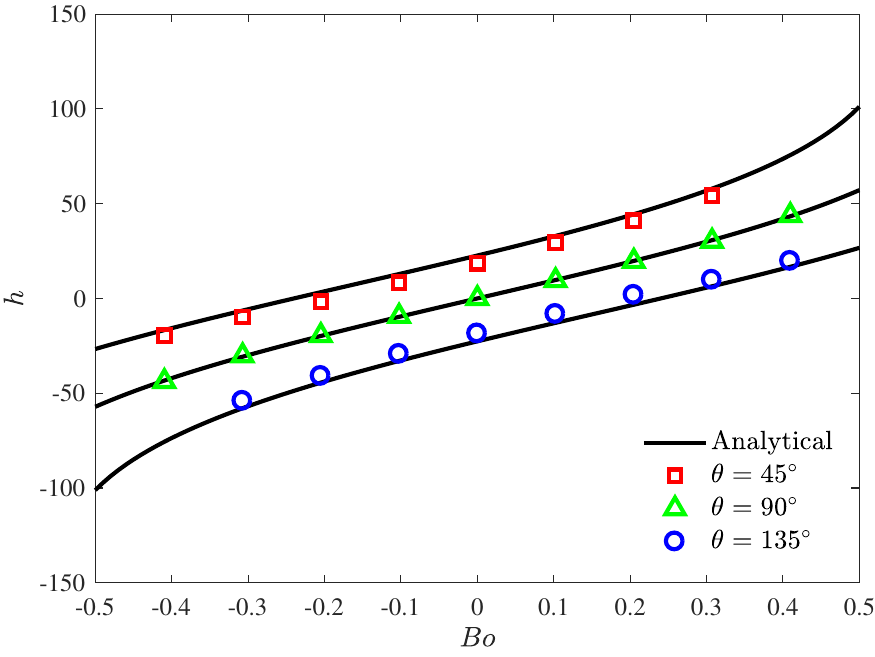}
	\caption{The downward distance $h$ from the fluid interface to the center of the particle.}
	\label{fig-1particle-Com}
\end{figure}

The fourth problem we consider is the balance of wetting particles on the liquid-liquid interface in presence of gravity. The computational domain is a square chamber with the length $L$, the periodic boundary condition is applied at left and right boundaries and no-flux boundary condition is adopted at upper and lower boundaries. The liquid-liquid interface is initialized at $y=L/2$, and the wetting particles with the same radius $R_p$ are half immersed in two fluids. To conduct a comparison with the theoretical solution \cite{Mino2020PRE}, the density ratio of two fluids is set to be the same ($\rho_1=\rho_2=\rho_f$). The particles will move to the equilibrium positions under the balance of gravity and surface tension, which can be depicted by the Bond number,
\begin{equation}
	Bo=\frac{\left(\rho_p-\rho_f\right)R_p^2|\mathbf{g}|}{\sigma}.
\end{equation}

We first study a single particle placed at the central of the liquid-liquid interface. As shown in Fig. \ref{fig-1particle}(c), the downward distance $h$ from the fluid interface at the left (or right) boundary to the center of the particle can be obtained from the following geometric relation,
\begin{equation}\label{singleH}
	h=R_p\cos\left(\theta-\alpha\right)+\frac{L}{\pi Bo}\left(1-\cos\alpha\right),
\end{equation}
where $\alpha$ is the slope angle of the fluid-fluid interface, and can be determined from the force balance \cite{Mino2020PRE},
\begin{equation}
	\frac{L}{2}=R_p\sin\left(\theta-\alpha\right)+\frac{L}{\pi Bo}\sin\alpha.
\end{equation} 
For this problem, we set $L=256$, $R_p=32$, $\rho_f=1$, $|\mathbf{g}|=2\times10^{-6}$, and the other parameters are the same as those in Section \ref{sec-porous}. We perform some simulations, and show the results in Fig. \ref{fig-1particle}. From this figure, one can find that the particle moves toward its mechanically equilibrium position in the gravitational field, and it moves upwards with $Bo<0$ but downwards when the gravity is larger than the buoyancy. Figure \ref{fig-1particle-Com} presents a comparison of the measured distance $h$ with the theoretical solution (\ref{singleH}), and an agreement between them is observed. 

\begin{figure}
	\centering
	\subfigure[$\rho_{p,1}=\rho_{p,2}=0.5$]{
		\begin{minipage}{0.32\linewidth}
			\centering
			\includegraphics[width=1.5in]{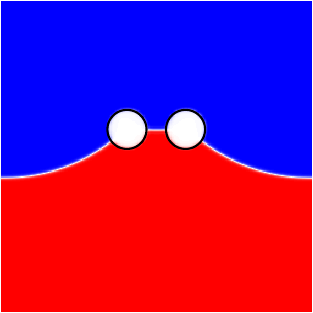}
		\end{minipage}
		\begin{minipage}{0.32\linewidth}
			\centering
			\includegraphics[width=1.5in]{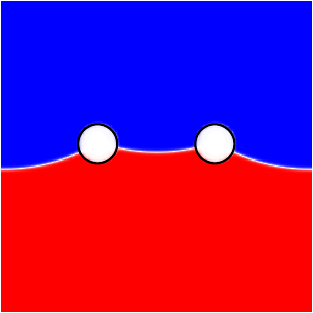}
		\end{minipage}
		\begin{minipage}{0.32\linewidth}
			\centering
			\includegraphics[width=1.5in]{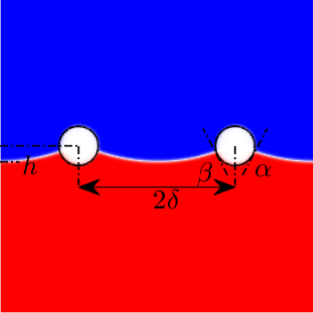}
	\end{minipage}}
	\subfigure[$\rho_{p,1}=\rho_{p,2}=1.5$]{
		\begin{minipage}{0.32\linewidth}
			\centering
			\includegraphics[width=1.5in]{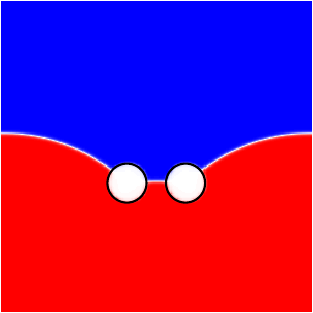}
		\end{minipage}
		\begin{minipage}{0.32\linewidth}
			\centering
			\includegraphics[width=1.5in]{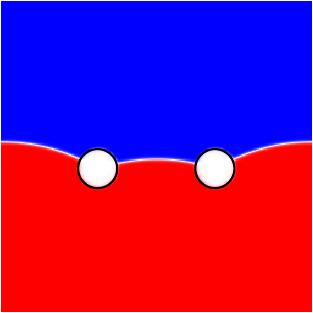}
		\end{minipage}
		\begin{minipage}{0.32\linewidth}
			\centering
			\includegraphics[width=1.5in]{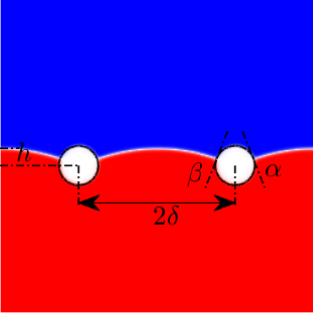}
	\end{minipage}}
	\subfigure[$\rho_{p,1}=1.5$, $\rho_{p,2}=0.5$]{
		\begin{minipage}{0.32\linewidth}
			\centering
			\includegraphics[width=1.5in]{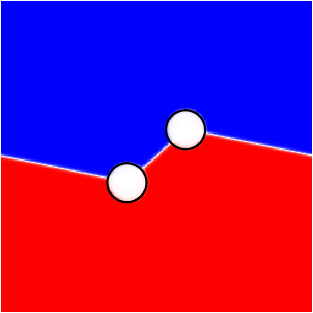}
		\end{minipage}
		\begin{minipage}{0.32\linewidth}
			\centering
			\includegraphics[width=1.5in]{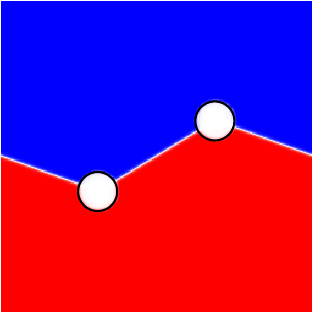}
		\end{minipage}
		\begin{minipage}{0.32\linewidth}
			\centering
			\includegraphics[width=1.5in]{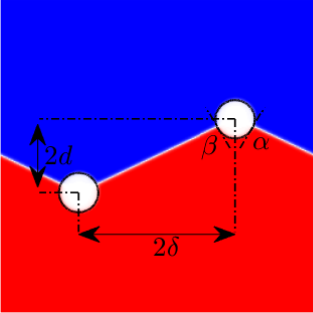}
	\end{minipage}}
	\caption{The mechanical equilibrium states of the two particles trapped at fluid-fluid interface in presence of gravity with distance $\delta=24\Delta x,48\Delta x,64\Delta x$ (from left to right columns).}
	\label{fig-2particle}
\end{figure}
\begin{figure}
	\centering
	\subfigure[]{
		\begin{minipage}{0.48\linewidth}
			\centering
			\includegraphics[width=3.0in]{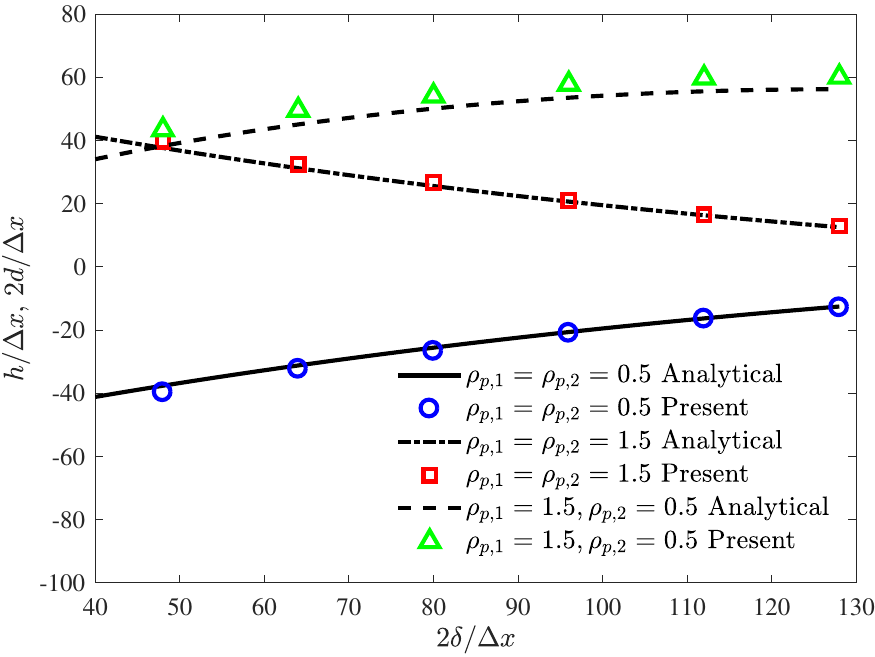}
	\end{minipage}}
	\subfigure[]{
		\begin{minipage}{0.48\linewidth}
			\centering
			\includegraphics[width=3.0in]{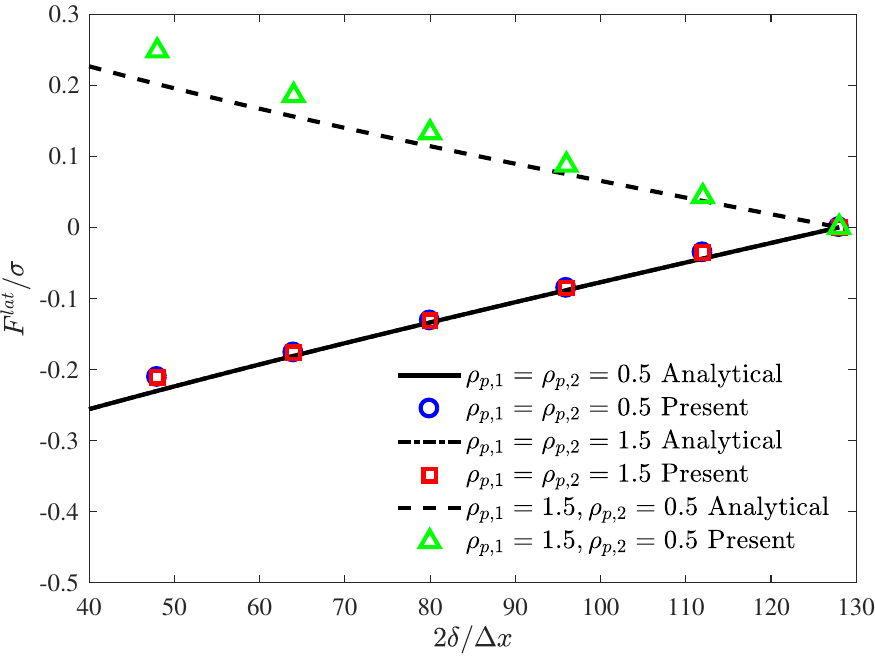}
	\end{minipage}}
	\caption{Comparisons of numerical predictions and the theoretical solutions for two particles trapped at the fluid-fluid interface under gravity [(a) the downward distance $h$ from the interface to the center of the particles with the same density and the vertical distance $2d$ between two particles, (b) the repulsive or attractive force $F^{lat}$ between two particles in the horizontal direction].}
	\label{fig-2particleG-Com}
\end{figure}

We also consider the case of two symmetrical particles with the horizontal distance being $2\delta$. The particle radius is set as $R_p=16$, and $D=3\Delta x$ is used for these two small particles. The theoretical solution of distance from the interface to the center of the particle $h$ (or the vertical distance of the two particles $d$) can also be derived from the force balance equations in horizontal and vertical directions (for more details, see Ref. \cite{Mino2020PRE}),
\begin{subequations}
	\begin{equation}
		F^{lat}=\sigma\left(\cos\alpha-\cos\beta\right)+R_p\left[-\sin\left(\frac{\pi}{2}-\theta+\alpha\right)+\sin\left(\frac{\pi}{2}-\theta+\beta\right)\right]\Delta p,
	\end{equation}
	\begin{equation}
		\sigma\left(\sin\alpha+\sin\beta\right)+R_p\left[\cos\left(\frac{\pi}{2}-\theta+\alpha\right)+\cos\left(\frac{\pi}{2}-\theta+\beta\right)\right]-\pi\left(\rho_p-\rho_f\right)R_p^2|\mathbf{g}|=0,
	\end{equation}
\end{subequations}
where $\alpha$ is the slope angle of the wide-gap fluid-fluid interface at the point of contact between the interface and the particle surface, and $\beta$ is that of the narrow-gap fluid-fluid interface (see the right column in Fig. \ref{fig-2particle}). $F^{lat}>0$ and $F^{lat}<0$ represent repulsive and attractive forces between two particles in the horizontal direction, respectively. Here it should be noted that through imposing the force $F^{lat}$ in the horizontal direction, the two particles are only allowed to move in the vertical direction. Once they are allowed to move in all directions, only the case of $2\delta=L/2$ can reach the steady state without particles touching together. 

We present some results under different densities and initially horizontal distances of the two particles in Fig. \ref{fig-2particle}, and find that when $\rho_{p,1}=\rho_{p,2}$, the value of $h$ decreases with the increase of $\delta$, while for other cases, an opposite trend is observed. We also conduct some quantitative comparisons of distances ($h$ and $d$) and force ($F^{lat}$) in Fig. \ref{fig-2particleG-Com}, and the numerical results agree well with the theoretical solutions. These validations indicate that the present diffuse-interface model can accurately describe the gas-liquid-solid multiphase flows, and in the following, it will be applied to study more complex multiphase flows.

\subsection{Mixing of two immiscible fluids and particles in an annulus with an inner rotating cylinder}\label{sec-mix}
Rotary drums are widely used in a range of chemical engineering industries to deal with liquid-solid mixing \cite{Tang2022PT}. In this part, we consider a simple case of gas-liquid-particle mixing in an annular gap with an inner rotating cylinder. As illustrated in Fig. \ref{fig-rotationInit}, two immiscible fluids and 15 particles are initialized in an annulus with the outer radius $R_{out}=2.56$ and inner radius $R_{in}=0.64$, and the inner cylinder rotates with a rotating speed $\omega_{in}$ to mix them. To study this problem, a larger square region $[-2.56,2.56]\times[-2.56,2.56]$ is used as the computational domain, and no-flux boundary condition is imposed on all boundaries. To depict the outer neutral solid region and the inner neutral cylinder of the annulus, two additional phase-field variables ($\phi_{0,out}$ and $\phi_{0,in}$) are introduced, and the initialization of all the phase-field variables are given by
\begin{equation}
	\begin{aligned}
		&\phi_{0,out}\left(x,y\right)=\frac{1}{2}+\frac{1}{2}\tanh\frac{\sqrt{x^2+y^2}-R_{out}}{D/2},\\
		&\phi_{0,in}\left(x,y\right)=\frac{1}{2}+\frac{1}{2}\tanh\frac{R_{in}-\sqrt{x^2+y^2}}{D/2},\\
		&\phi_{0,k}\left(x,y\right)=\frac{1}{2}+\frac{1}{2}\tanh\frac{R_{p,k}-\sqrt{\left(x-x_{p,k}\right)^2+\left(y-y_{p,k}\right)^2}}{D/2},\\
		&\phi\left(x,y\right)=\left[1-\phi_{0,out}\left(x,y\right)-\phi_{0,in}\left(x,y\right)-\sum_{k=1}^{15}\phi_{0,k}\left(x,y\right)\right]\tanh\frac{-2y}{D},
	\end{aligned}
\end{equation} 
where $R_{p,k}$ and $\mathbf{x}_{p,k}=\left(x_{p,k},y_{p,k}\right)$ are the radius and initial position of the $k$-th particle, respectively. In our simulations, the gravity with $|\mathbf{g}|=0.01$ is imposed on particles, and the material properties of particles are randomly determined in the given regions, i.e., $R_{p,k}$ is selected between 0.2 and 0.3, the density of $k$-th particle is given between 0.95 and 1.05, and the contact angle $\theta_{k}$ of $k$-th particle is set between $45^\circ$ and $135^\circ$. Additionally, the physical parameters of two fluids are given as $\rho_1=\rho_2=1$, $\mu_1=\mu_2=0.1$, and $\sigma_{12}=0.001$. 

\begin{figure}
	\centering
	\includegraphics[width=1.5in]{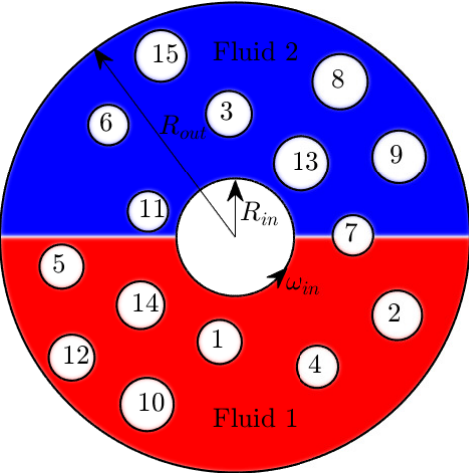}
	\caption{An annulus with the inner rotating cylinder.}
	\label{fig-rotationInit}
\end{figure}

In addition, to avoid the overlap between particles or between the particles and walls, the following short-range repulsive forces are considered \cite{Wan2006IJNMF},
\begin{subequations}
	\begin{equation}
		\mathbf{F}_{k,j}^{pp}=\begin{cases}
			0, &d_{kj}>R_{p,k}+R_{p,j}+\xi,\\
			\frac{1}{\epsilon}\left(\mathbf{x}_{p,k}-\mathbf{x}_{p,j}\right)\left(R_{p,k}+R_{p,j}+\xi-d_{kj}\right)^2, &R_{p,k}+R_{p,j}\leq d_{kj}\leq R_{p,k}+R_{p,j}+\xi,\\
			\frac{1}{\epsilon}\left(\mathbf{x}_{p,k}-\mathbf{x}_{p,j}\right)\left(R_{p,k}+R_{p,j}-d_{kj}\right), &d_{kj}\leq R_{p,k}+R_{p,j},\\
		\end{cases}
	\end{equation} 
	\begin{equation}
		\mathbf{F}_{k}^{pw}=\begin{cases}
			0, &d'_{k}>2R_{p,k}+\xi,\\
			\frac{1}{\epsilon/2}\left(\mathbf{x}_{p,k}-\mathbf{x}'_{p,k}\right)\left(2R_{p,k}+\xi-d'_{k}\right)^2, &2R_{p,k}\leq d'_{k}\leq 2R_{p,k}+\xi,\\
			\frac{1}{\epsilon/2}\left(\mathbf{x}_{p,k}-\mathbf{x}'_{p,k}\right)\left(2R_{p,k}-d'_{k}\right), &d'_{k}\leq 2R_{p,k},\\
		\end{cases}
	\end{equation}
\end{subequations}
where $\mathbf{x}'_{p,k}$ is the position of the nearest imaginary particle on the boundary wall, $d_{kj}=|\mathbf{x}_{p,k}-\mathbf{x}_{p,j}|$ and $d'_{k}=|\mathbf{x}_{p,k}-\mathbf{x}'_{p,k}|$. $\xi$ is the threshold of the repulsion forces, and $\epsilon$ is the stiff parameter of the scheme. 
 
We perform some simulations with $\Delta x=0.01$, $\Delta t=0.001$, $M=0.01$, $D=5\Delta x$, $\xi=D/2$, $\epsilon=\Delta x/2$, and show the mixing processes at some specific moments in Fig. \ref{fig-rotation} where $\omega_{in}=0.1$. From this figure, one can observe that the two layered fluids mix into each other as time goes on, and the particles also move in the annular gap under the effects of rotating speed, gravity, surface tension and repulsive forces. Additionally, the trajectories of the particles in some specific time intervals are plotted in Fig. \ref{fig-rotationXY}. In this figure, the overall movement trends of particles are rotating in time, but there are many special changes at some moments. For instance, in Fig. \ref{fig-rotationXY}(a), the moving direction of the 6-th particle is changed by the repulsive force of the outer wall, the direction of the 3-rd particle is turned by the repulsive forces between particles, and the 14-th particle becomes close to the 11-th particle because of the low-pressure wake created by the 11-th particle.  

To quantitatively describe the mixing degree, we calculate the information entropy of phase-field variable by the following formula,
\begin{equation}
	H\left(\phi\right)=-\sum_np\left(\phi_n\right)\log_2p\left(\phi_n\right),
\end{equation} 
where $p\left(\phi_n\right)$ is the probability of the $n$-th value range of $\phi$. As shown in Fig. \ref{fig-rotationH}, the value of $H\left(\phi\right)$ increases with the mixing of fluids and particles at first, and then oscillates near a certain value when the rotating speed reaches a quasi-steady state. 
One can also find from this figure that the value of $H\left(\phi\right)$ also increases with the increase of the rotating speed due to a larger driving force.

\begin{figure}
	\centering
	\subfigure[$t=20$]{
		\begin{minipage}{0.24\linewidth}
			\centering
			\includegraphics[width=1.5in]{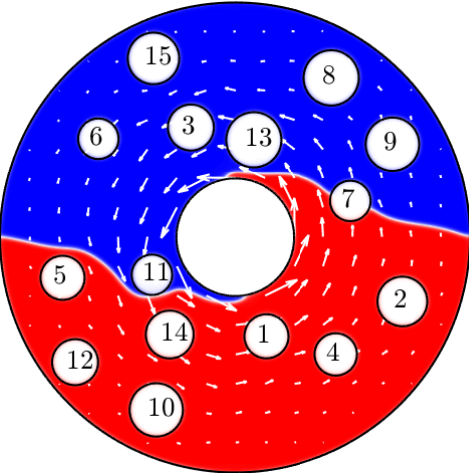}
		\end{minipage}}
	\subfigure[$t=40$]{
		\begin{minipage}{0.24\linewidth}
			\centering
			\includegraphics[width=1.5in]{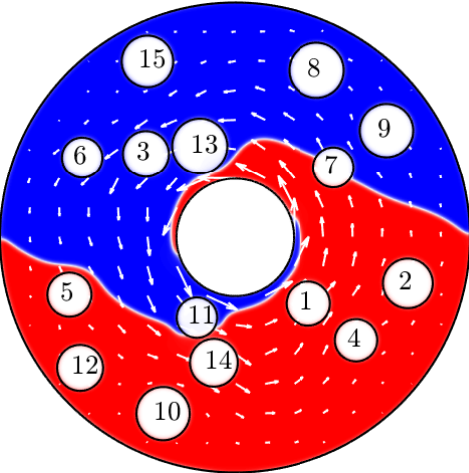}
		\end{minipage}}
	\subfigure[$t=60$]{
		\begin{minipage}{0.24\linewidth}
			\centering
			\includegraphics[width=1.5in]{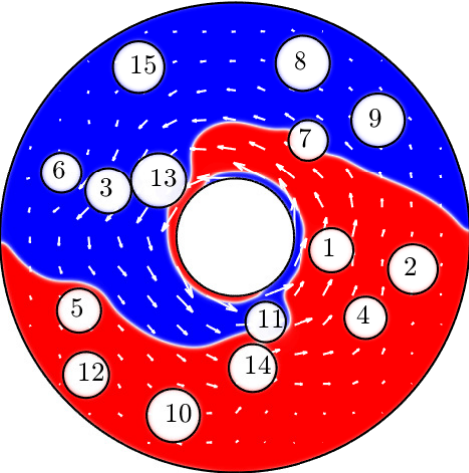}
		\end{minipage}}
	\subfigure[$t=80$]{
		\begin{minipage}{0.24\linewidth}
			\centering
			\includegraphics[width=1.5in]{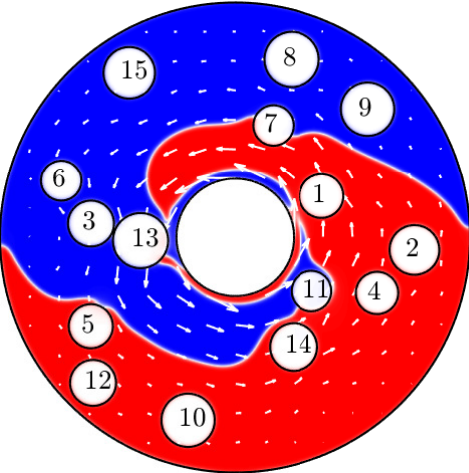}
	\end{minipage}}

	\subfigure[$t=100$]{
		\begin{minipage}{0.24\linewidth}
			\centering
			\includegraphics[width=1.5in]{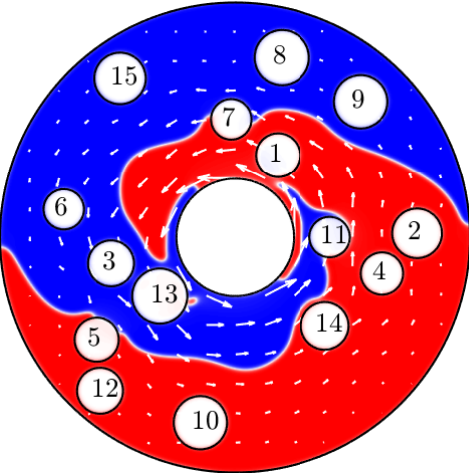}
	\end{minipage}}
	\subfigure[$t=200$]{
		\begin{minipage}{0.24\linewidth}
			\centering
			\includegraphics[width=1.5in]{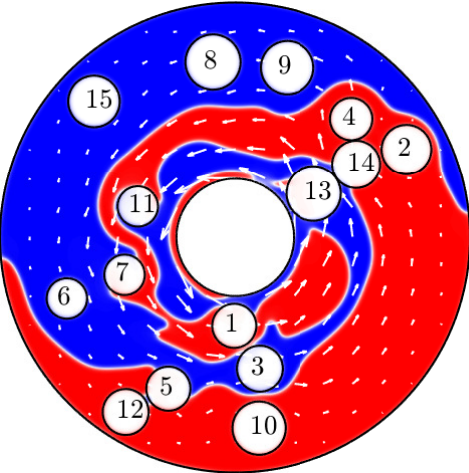}
	\end{minipage}}
	\subfigure[$t=300$]{
		\begin{minipage}{0.24\linewidth}
			\centering
			\includegraphics[width=1.5in]{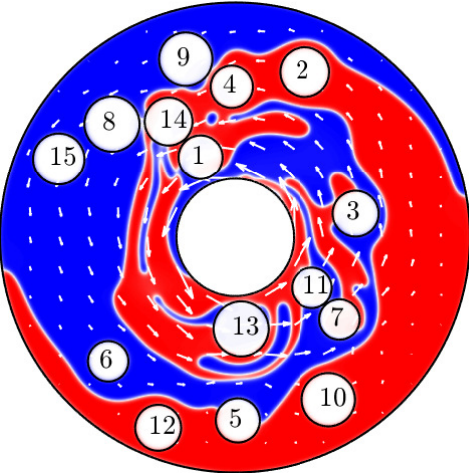}
	\end{minipage}}
	\subfigure[$t=400$]{
		\begin{minipage}{0.24\linewidth}
			\centering
			\includegraphics[width=1.5in]{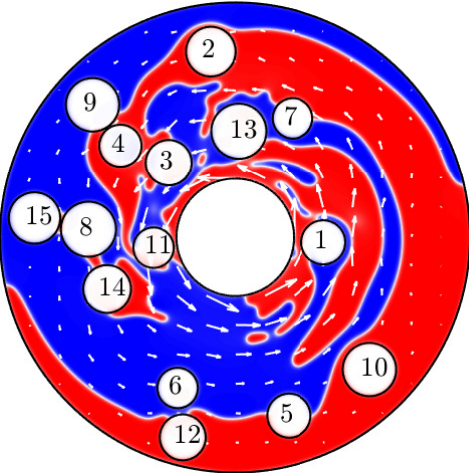}
	\end{minipage}}

	\subfigure[$t=500$]{
		\begin{minipage}{0.24\linewidth}
			\centering
			\includegraphics[width=1.5in]{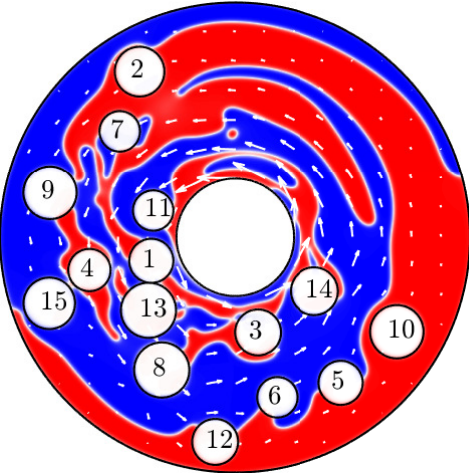}
	\end{minipage}}
	\subfigure[$t=1000$]{
		\begin{minipage}{0.24\linewidth}
			\centering
			\includegraphics[width=1.5in]{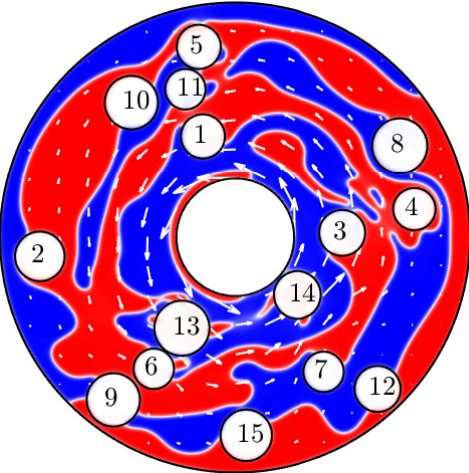}
	\end{minipage}}
	\subfigure[$t=1500$]{
		\begin{minipage}{0.24\linewidth}
			\centering
			\includegraphics[width=1.5in]{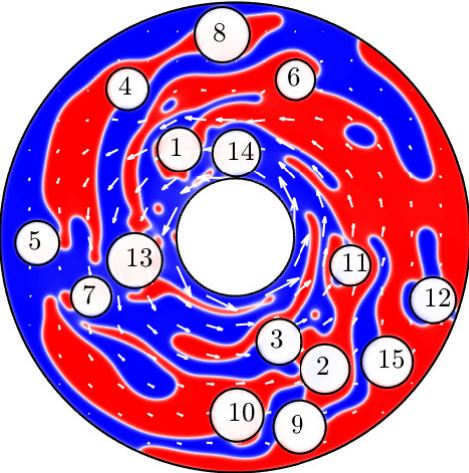}
	\end{minipage}}
	\subfigure[$t=2000$]{
		\begin{minipage}{0.24\linewidth}
			\centering
			\includegraphics[width=1.5in]{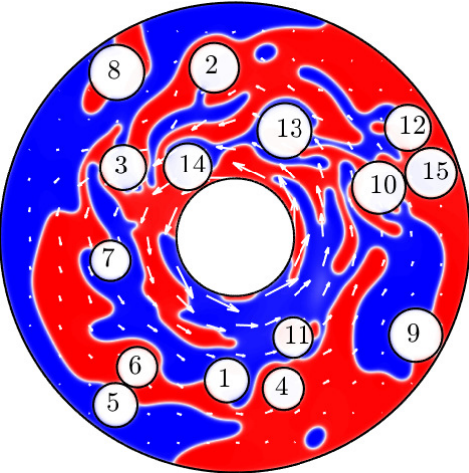}
	\end{minipage}}
	\caption{The mixing of the fluids and particles in an annulus with inner rotating cylinder ($\omega_{in}=0.1$).}
	\label{fig-rotation}
\end{figure}
\begin{figure}
	\centering
	\subfigure[]{
		\begin{minipage}{0.48\linewidth}
			\centering
			\includegraphics[width=2.5in]{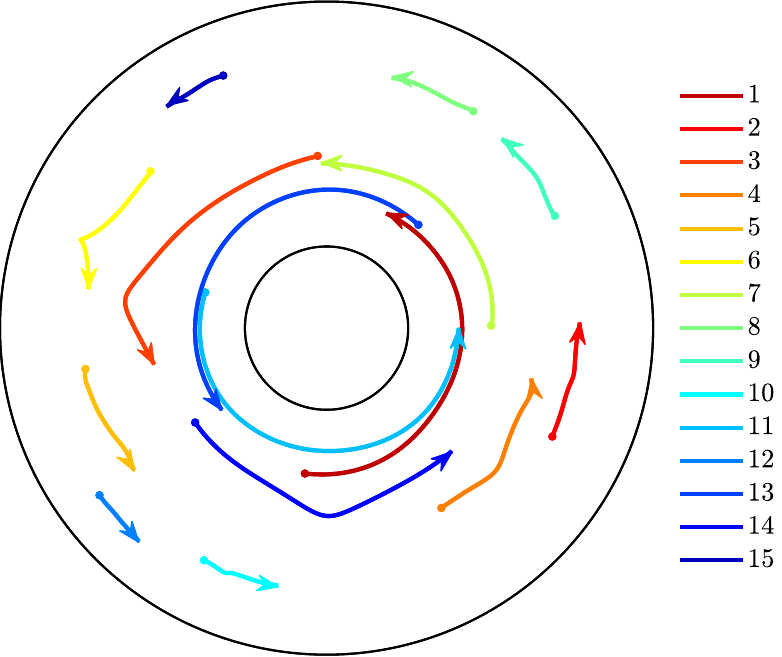}
	\end{minipage}}
	\subfigure[]{
		\begin{minipage}{0.48\linewidth}
			\centering
			\includegraphics[width=2.5in]{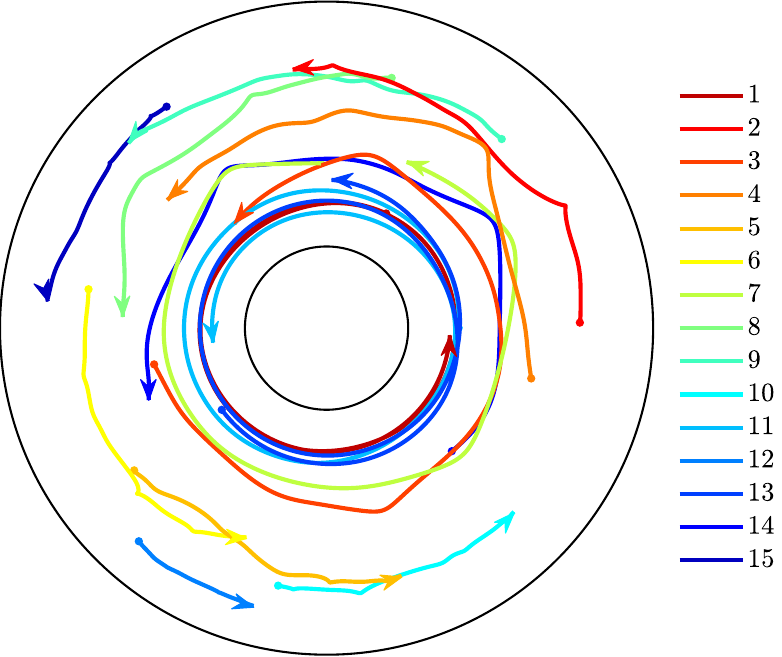}
	\end{minipage}}

	\subfigure[]{
		\begin{minipage}{0.48\linewidth}
			\centering
			\includegraphics[width=2.5in]{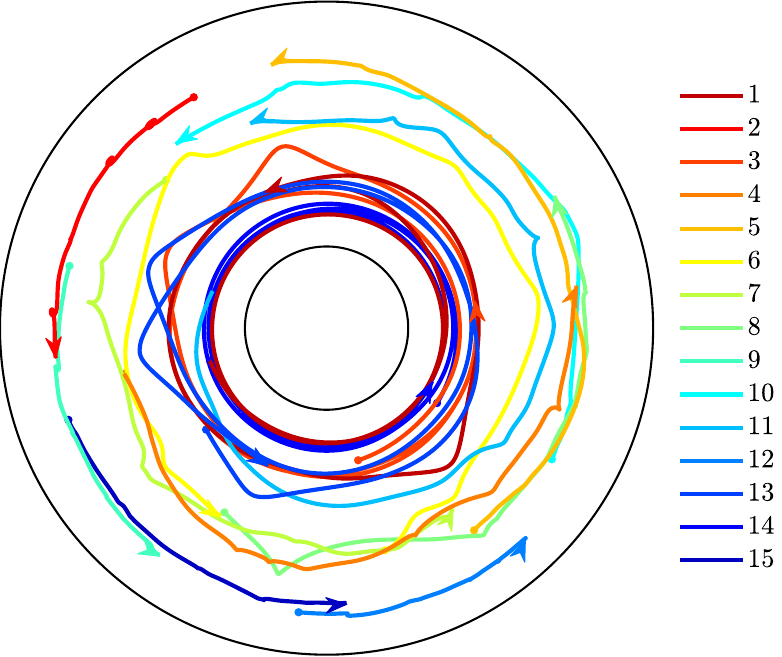}
	\end{minipage}}
	\subfigure[]{
		\begin{minipage}{0.48\linewidth}
			\centering
			\includegraphics[width=2.5in]{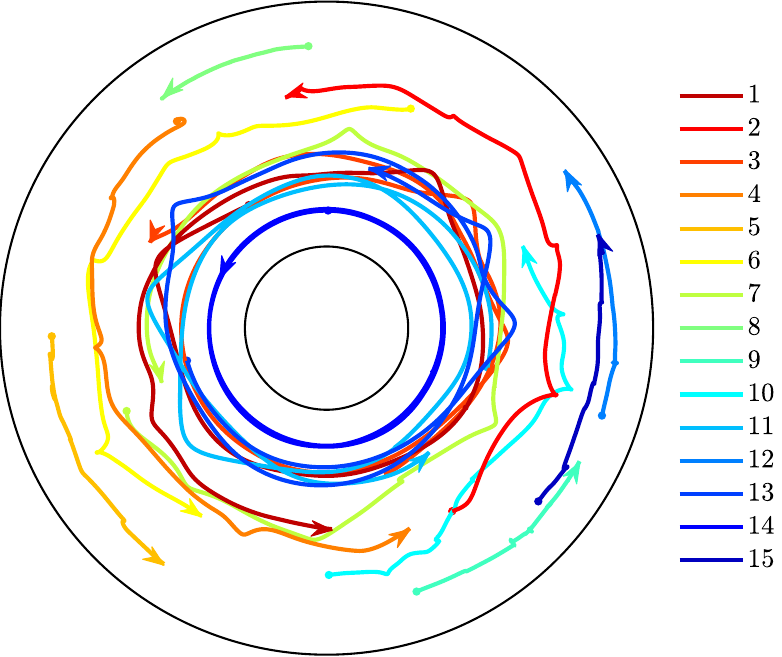}
	\end{minipage}}
	\caption{The trajectories of the particles in some specific time intervals [(a) $0\leq t\leq100$, (b) $100\leq t\leq400$, (c) $500\leq t\leq1000$, and (d) $1500\leq t\leq2000$].}
	\label{fig-rotationXY}
\end{figure}
\begin{figure}
	\centering
	\includegraphics[width=3.5in]{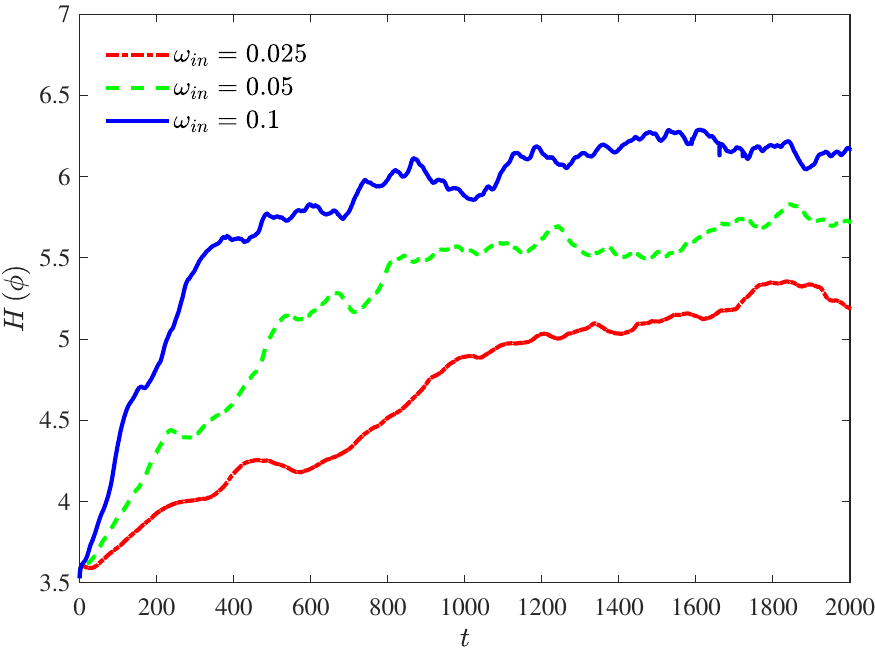}
	\caption{The evolution of information entropy of the phase-field variable $\phi$.}
	\label{fig-rotationH}
\end{figure}

\subsection{Displacement of immiscible fluids and particles in a complex channel}
\begin{figure}
	\centering
	\includegraphics[width=4.0in]{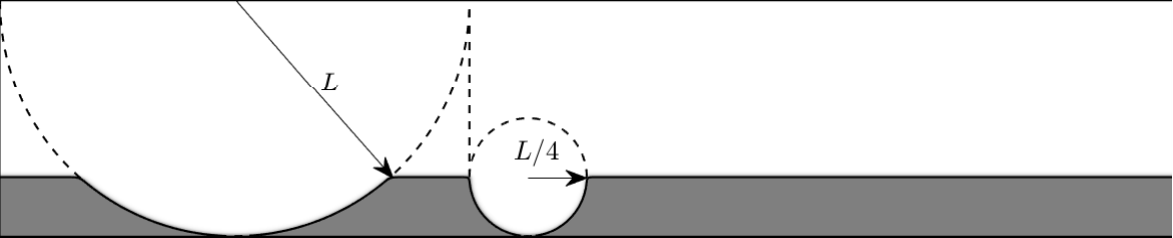}
	\caption{Schematic of the complex channel, the gray region represents the solid wall.}
	\label{fig-cavityInit}
\end{figure}
The displacement of immiscible fluids is an important issue for industry applications, such as the geological sequestration of carbon dioxide and oil recovery \cite{Alvarado2010,Fan2019Energy}. In recent years, the simplified problem of displacement in a cavity has been widely studied \cite{Zhou2010CES,Liang2015IJMPC,Lou2017IJMPC}, but the effect of free solid objects on the displacement has not been considered in these previous works, which may also play an important role in the above-mentioned applications. Here we assume the free solid objects to be the circular particles, and investigate the immiscible displacement in a cavity structure. The schematic of the problem is shown in Fig. \ref{fig-cavityInit} where the width and length of the computational domain are $L=2.56$ and $5L$, respectively. The periodic boundary condition is applied in the horizontal direction, while the no-flux boundary condition is imposed on the neutral bottom and top walls.  

Initially, 16 particles are placed in the channel with random radius ($0.16\leq R_{p,k}\leq0.2$), density ($0.95\leq\rho_{p,k}\leq1.05$) as well as contact angle ($80^\circ\leq\theta_{p,k}\leq100^\circ$), and most of them are placed in the left cavity. In addition, a droplet of the displaced fluid 1 with $\rho_1=1$ and $\mu_1=0.1$ is located in the right cavity, and other parts of the domain are filled by the driving fluid 2 with $\rho_2=1$ and $\mu_2=0.01$. To achieve the above settings, the phase-field variables are initialized by
\begin{equation}
	\begin{aligned}
		&\phi_{0,bot}\left(x,y\right)=\left(\frac{1}{2}+\frac{1}{2}\tanh\frac{L/4-y}{D/2}\right)\left[\frac{1}{2}+\frac{1}{2}\tanh\frac{\sqrt{\left(x-L\right)^2+\left(y-L\right)^2}-L}{D/2}\right]\left[\frac{1}{2}+\frac{1}{2}\tanh\frac{\sqrt{\left(x-5L/4\right)^2+\left(y-L/4\right)^2}-L/4}{D/2}\right],\\
		&\phi_{0,k}\left(x,y\right)=\frac{1}{2}+\frac{1}{2}\tanh\frac{R_{p,k}-\sqrt{\left(x-x_{p,k}\right)^2+\left(y-y_{p,k}\right)^2}}{D/2},\\
		&\phi\left(x,y\right)=\left[1-\phi_{0,bot}\left(x,y\right)-\sum_{k=1}^{16}\phi_{0,k}\left(x,y\right)\right]\tanh\frac{L/4-\sqrt{\left(x-5L/4\right)^2+\left(y-L/4\right)^2}}{D/2}.
	\end{aligned}
\end{equation}
To drive the fluids and particles, a body force $\mathbf{F}_b=\left(0.001,0\right)$ along the positive $x$-direction is imposed on the fluids, and the effect of the gravity is neglected in this case. In the following simulations, other parameters are set the same as those in Section \ref{sec-mix}. 

We consider three different cases of the displacement problem, case A: single fluid and particles, case B: only immiscible fluids, and case C: immiscible fluids and particles. As shown in Fig. \ref{fig-cavity-com}, in the early stage of displacement when the particles do not contact with the fluid 1, the particles show the same trajectories for cases A and C, while the fluid 1 shows a slower displacement for case C due to the hindrance of the particles in the upstream of the droplet, compared to case B. On the other hand, when the particles interact with the fluid 1 (see case C), it seems that more fluid 1 is displaced from the cavity, and some particles are trapped by viscous force and surface tension, and thus move slower in the channel.   

To give a quantitative comparison of the displacement efficiency, we define $De$ to be the ratio of the volume of the fluid 1 out the right cavity ($x<2L$ and $x>5L/2$) to the total volume of fluid 1, and plot the results in Fig. \ref{fig-cavity-Q}(a). From this figure, one can find that when $t<250$, the presence of particles in the upstream reduces the fluid velocity, which leads to a lower displacement efficiency. However, when some particles catch up with and exceed the fluid 1, they promote some fluid 1 out of the cavity because of the effects of FSI and wettability, and thus the value of $De$ becomes larger than that of the case without particles. In addition, the horizontal positions of some typical particles of cases A and C are shown in Fig. \ref{fig-cavity-Q}(b). In this figure, the 15-th particle shows the same position for two cases since it does not interact with fluid 1. However, the trajectories of the 12-th and 16-th particles have some differences in the later stage: the 16-th particle is trapped by the viscous resistance of fluid 1, while the 12-th particle moves slowly due to a small driving velocity caused by the obstruction of particles in the downstream. These results illustrate that the presence of the particles can improve the displacement efficiency to a certain extent.     

\begin{figure}
	\centering
	\subfigure[$t=0$]{
		\begin{minipage}{0.33\linewidth}
			\centering
			\includegraphics[width=2.0in]{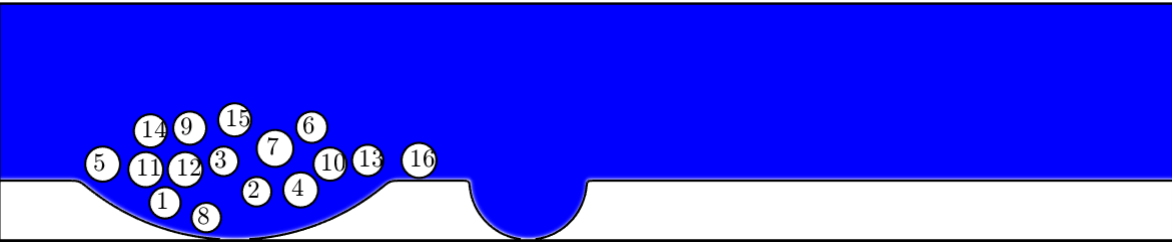}
		\end{minipage}
		\begin{minipage}{0.33\linewidth}
			\centering
			\includegraphics[width=2.0in]{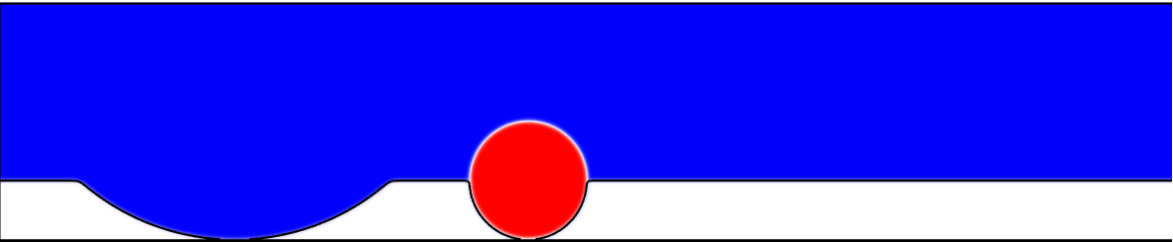}
		\end{minipage}
		\begin{minipage}{0.33\linewidth}
			\centering
			\includegraphics[width=2.0in]{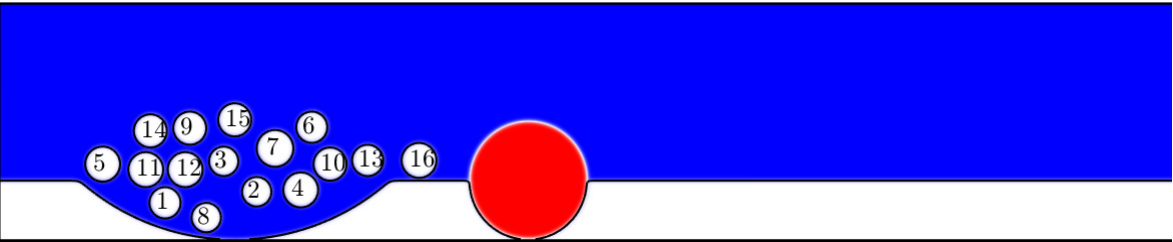}
		\end{minipage}}
	
	\subfigure[$t=50$]{
		\begin{minipage}{0.33\linewidth}
			\centering
			\includegraphics[width=2.0in]{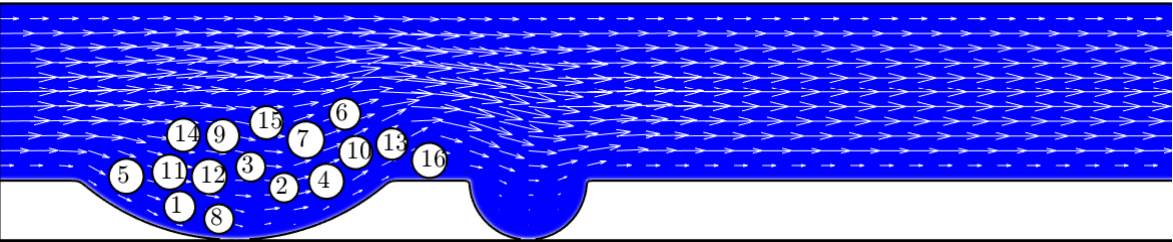}
		\end{minipage}
		\begin{minipage}{0.33\linewidth}
			\centering
			\includegraphics[width=2.0in]{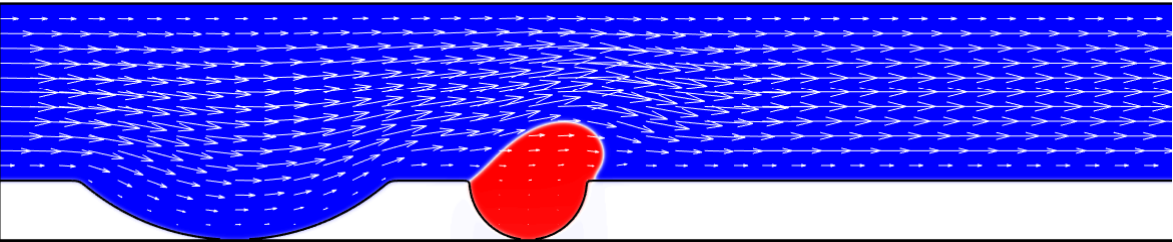}
		\end{minipage}
		\begin{minipage}{0.33\linewidth}
			\centering
			\includegraphics[width=2.0in]{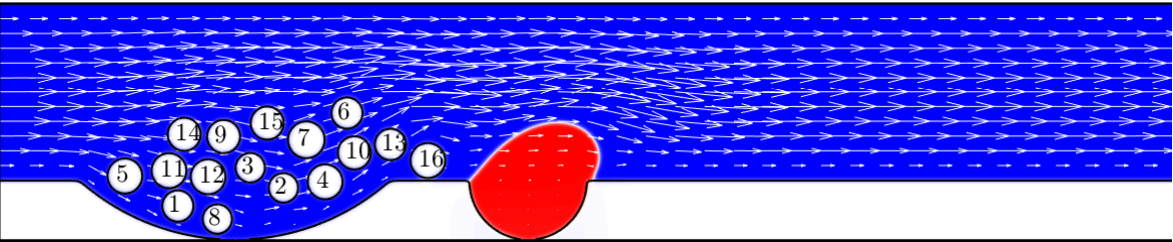}
	\end{minipage}}
	
	\subfigure[$t=100$]{
		\begin{minipage}{0.33\linewidth}
			\centering
			\includegraphics[width=2.0in]{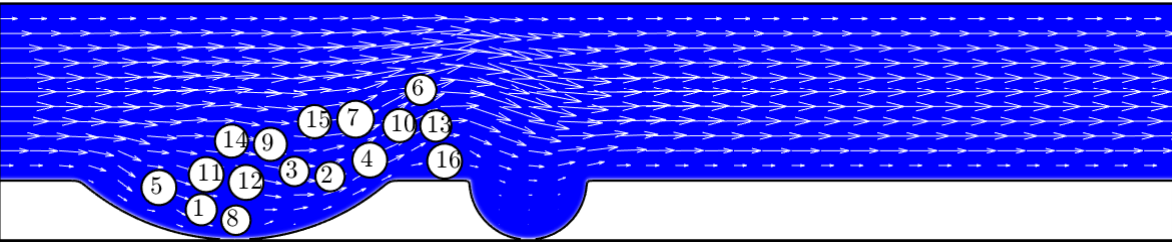}
		\end{minipage}
		\begin{minipage}{0.33\linewidth}
			\centering
			\includegraphics[width=2.0in]{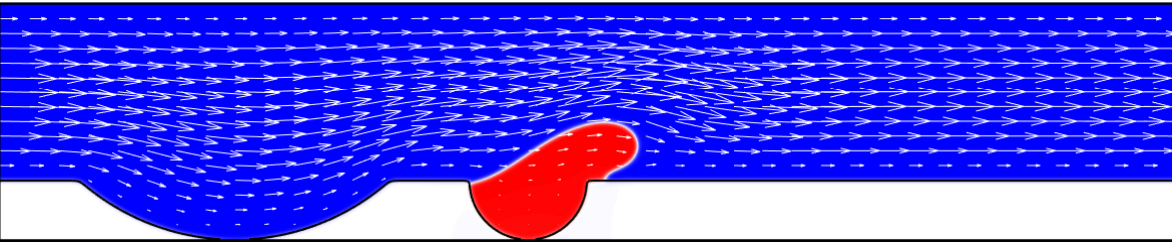}
		\end{minipage}
		\begin{minipage}{0.33\linewidth}
			\centering
			\includegraphics[width=2.0in]{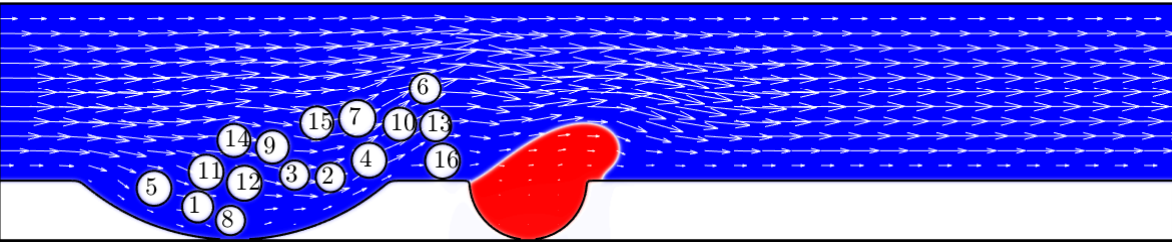}
	\end{minipage}}

	\subfigure[$t=200$]{
		\begin{minipage}{0.33\linewidth}
			\centering
			\includegraphics[width=2.0in]{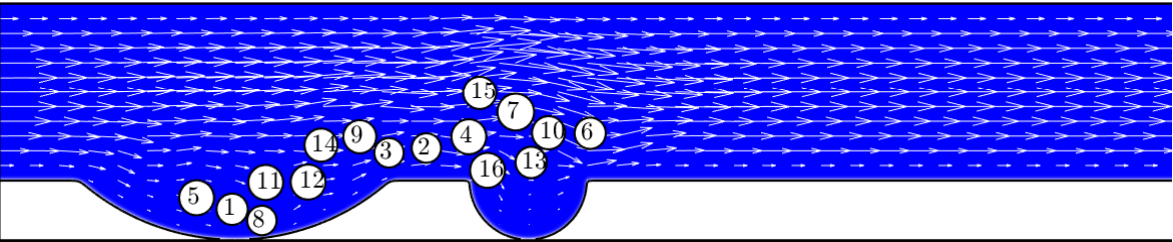}
		\end{minipage}
		\begin{minipage}{0.33\linewidth}
			\centering
			\includegraphics[width=2.0in]{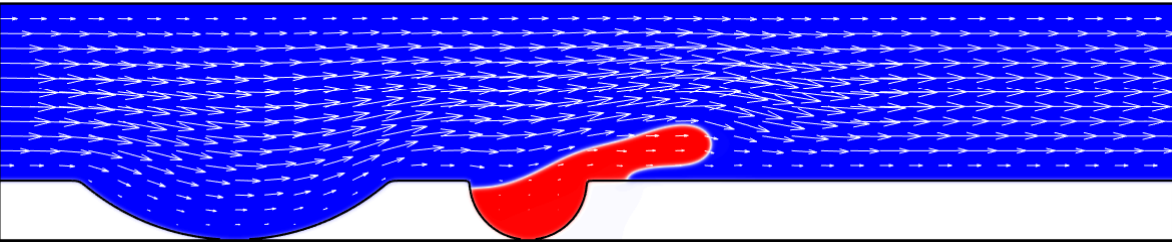}
		\end{minipage}
		\begin{minipage}{0.33\linewidth}
			\centering
			\includegraphics[width=2.0in]{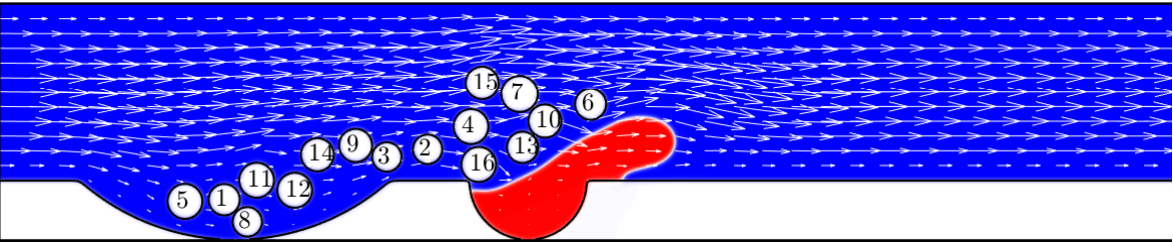}
	\end{minipage}}

	\subfigure[$t=300$]{
		\begin{minipage}{0.33\linewidth}
			\centering
			\includegraphics[width=2.0in]{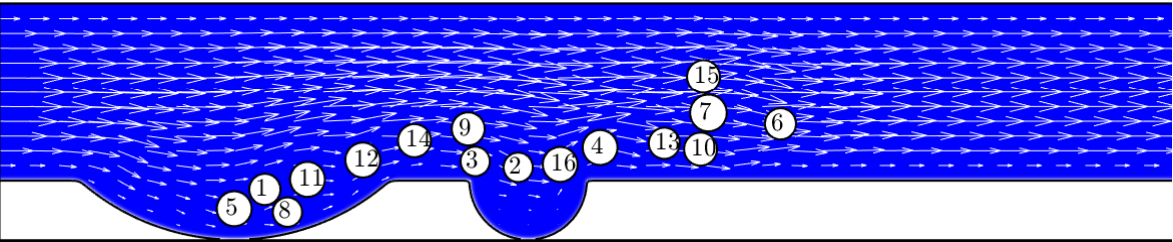}
		\end{minipage}
		\begin{minipage}{0.33\linewidth}
			\centering
			\includegraphics[width=2.0in]{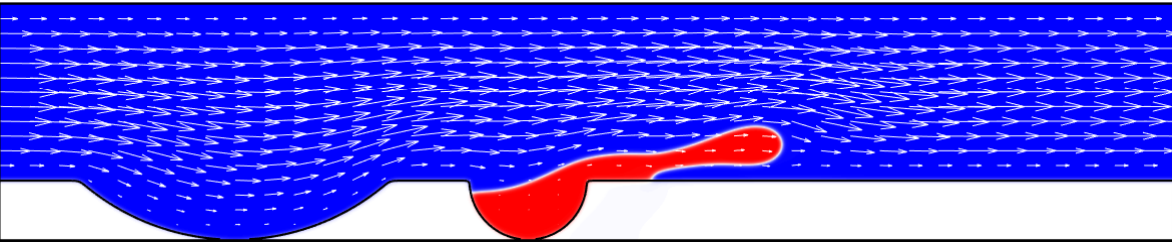}
		\end{minipage}
		\begin{minipage}{0.33\linewidth}
			\centering
			\includegraphics[width=2.0in]{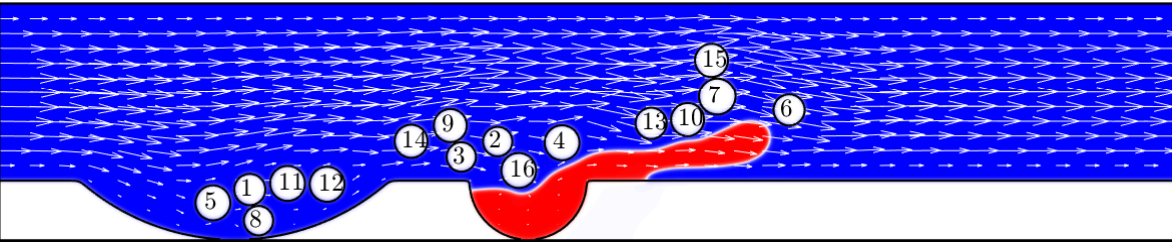}
	\end{minipage}}

	\subfigure[$t=450$]{
		\begin{minipage}{0.33\linewidth}
			\centering
			\includegraphics[width=2.0in]{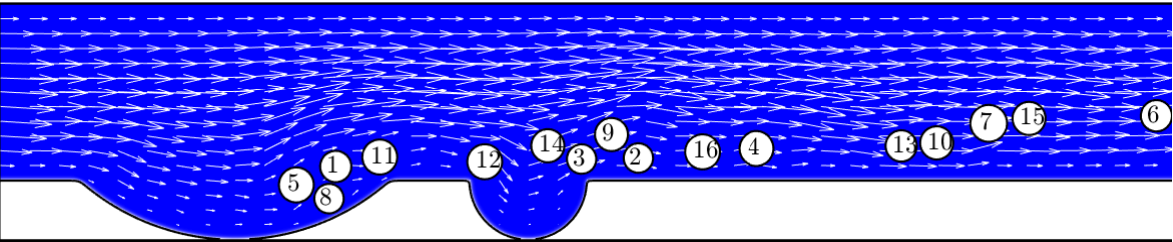}
		\end{minipage}
		\begin{minipage}{0.33\linewidth}
			\centering
			\includegraphics[width=2.0in]{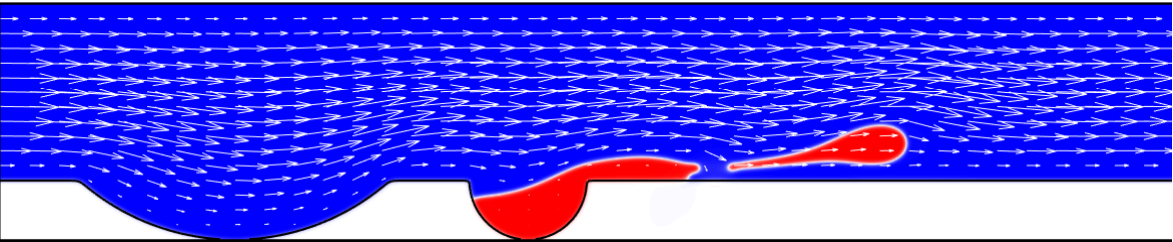}
		\end{minipage}
		\begin{minipage}{0.33\linewidth}
			\centering
			\includegraphics[width=2.0in]{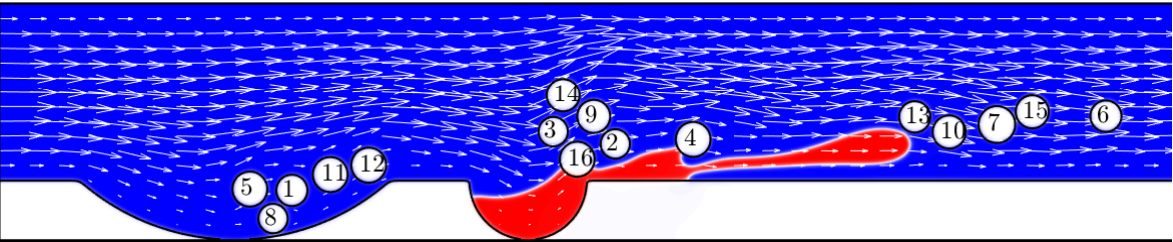}
	\end{minipage}}
	\caption{The dynamics of displacement in a complex channel (cases A, B and C from left to right columns).}
	\label{fig-cavity-com}
\end{figure}
\begin{figure}
	\centering
	\subfigure[]{
		\begin{minipage}{0.48\linewidth}
			\centering
			\includegraphics[width=3.0in]{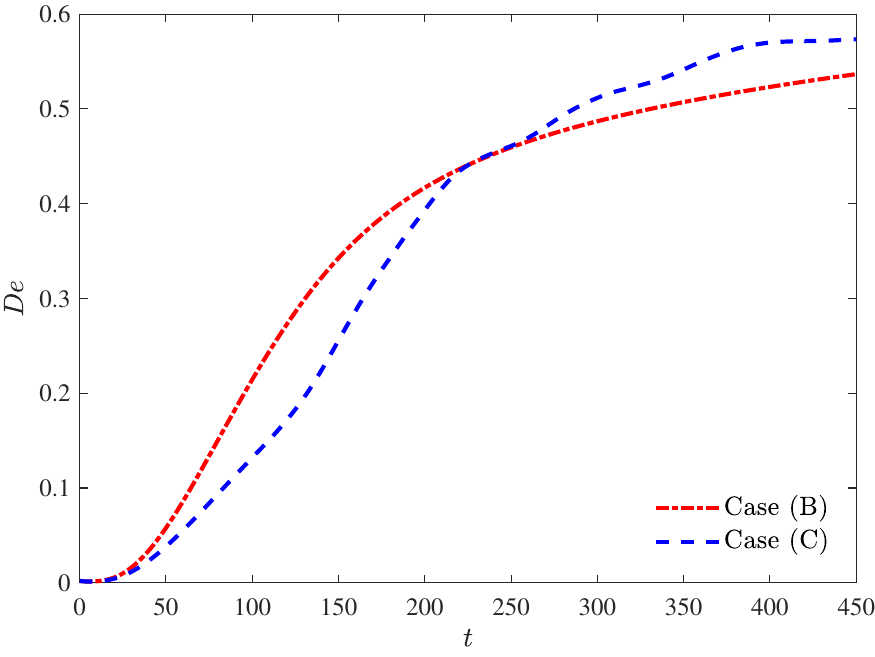}
	\end{minipage}}
	\subfigure[]{
		\begin{minipage}{0.48\linewidth}
			\centering
			\includegraphics[width=3.0in]{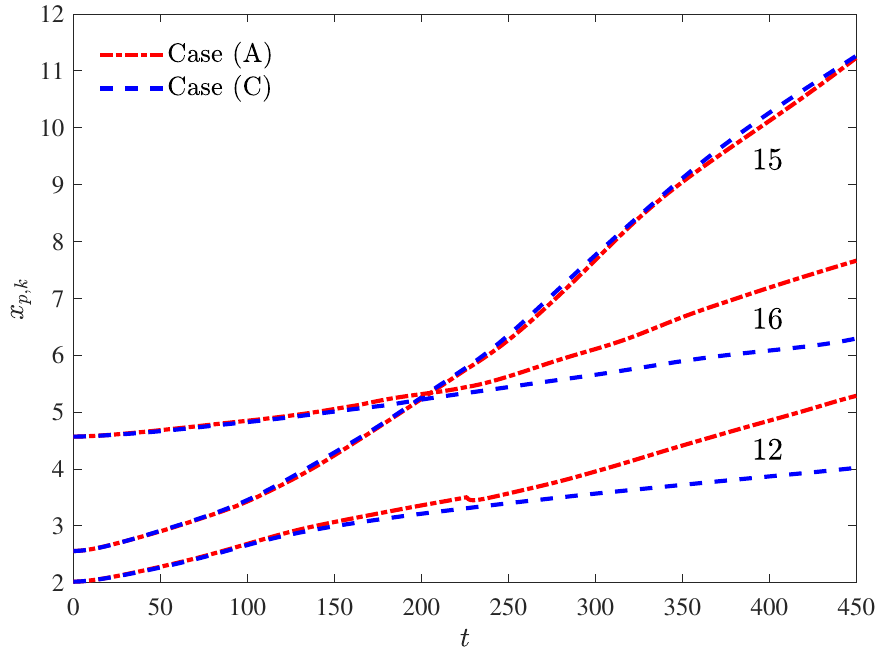}
	\end{minipage}}
	\caption{Time histories of some quantities [(a) displacement efficiency, (b) horizontal positions of some particles].}
	\label{fig-cavity-Q}
\end{figure}

\section{Conclusions}\label{Conclusions}
In this paper, a thermodynamically consistent and conservative diffuse-interface model for gas-liquid-solid multiphase flows is proposed based on a ternary phase-field model. In the present diffuse-interface model, a novel free energy including the original part for two-phase flows, the penalty term in solid phase as well as the term reflecting the wettability on the solid surface is also developed, and under some reasonable approximations, we can show that the last term is consistent with the wall energy in the classic phase-field model for two-phase flows. In addition, a consistent and conservative NS equations is also constructed for the fluid flows, in which the volume fraction of solid phase is introduced to depict the FSI and the high viscosity in solid phase. Furthermore, the total energy of the present phase-field-NS system is also proved to be dissipative for the two-phase flows. We note that the present diffuse-interface model not only can be used to describe the two-phase flows in complex geometries, but also can be applied to study the gas-liquid-particle multiphase flows by including the governing equations of particle motion. To solve the proposed diffuse-interface model, a general LB method is also developed where the scale factor $d_0$, instead of the fixed value in classical LB methods, can be adjusted to improve the numerical stability. The accuracy and the energy dissipation of the diffuse-interface model for gas-liquid-solid multiphase flows are first tested by several benchmark problems, then some multiphase flows in complex geometries with complex interfacial dynamics and FSI are further considered, and the numerical results demonstrate the good capability of the present diffuse-interface model.     

\section*{Acknowledgments}
This research was supported by the National Natural Science Foundation of China (Grants No. 12072127 and No 51836003), and the Interdiciplinary Research Program of HUST (2023JCYJ002). The computation was completed on the HPC Platform of Huazhong University of Science and Technology.
\appendix

\bibliographystyle{elsarticle-num} 
\bibliography{references}

\end{document}